\documentclass[conference]{IEEEtran}

\IEEEoverridecommandlockouts
\usepackage{amsmath,amssymb,amsthm,amsfonts}
\usepackage{graphicx}
\usepackage{textcomp}
\usepackage{xcolor}
\usepackage{pgfplots}
\usepackage{subfig}

\usepackage[adversary, complexity,  probability]{cryptocode}

\usepackage[colorlinks=true, citecolor=blue]{hyperref}
\usepackage[strings]{underscore} 
 \usepackage{breakcites}
  \usepackage{multirow}
\usepackage{diagbox}
  
\usepackage{listings}     
    
  \newcommand{\append}{Appendix}
    
\newcommand{\ZZ}{\mathbb{Z}}
\newcommand{\NN}{\mathbb{N}}

\newcommand{\BB}{\mathbb{B}}
\newcommand{\RR}{\mathbb{R}}

\newcommand{\R}{\mathcal{R}}
 
\renewcommand{\S}{\mathcal{S}}

\newcommand{\ma}[1]{\mathbf{#1}} 
\renewcommand{\vec}[1]{\mathbf{#1}}

\newcommand{\fr}[1]{\mathfrak{#1}}

\newcommand{\round}[1]{\lfloor #1 \rceil}

\newcommand{\roundUp}[1]{\lceil #1 \rceil}

\newcommand{\sample}{\leftarrow_{\$}} 
\newcommand{\exec}{\leftarrow} 
\newcommand{\tensor}{\otimes}

\newcommand{\extProd}{\mathsf{extProd}}

\newcommand{\CMux}{\mathsf{CMux}}
\renewcommand{\ge}{>}

\newcommand{\coefs}{\mathsf{Coefs}}

\newcommand{\cyclo}{\Phi}
\newcommand{\sgn}{\mathsf{sgn}}
\newcommand{\NegSgn}{\mathsf{NSgn}}
\newcommand{\norm}[1]{||#1||}
\newcommand{\abs}[1]{|#1|}

\newcommand{\Var}{\mathsf{Var}}
\newcommand{\stddev}{\mathsf{stddev}}
\newcommand{\E}{\mathsf{E}}
\newcommand{\erf}{\mathsf{erf}}

\newcommand{\hamming}{\mathsf{Ham}}

\newcommand{\LWE}{\mathsf{LWE}}
\newcommand{\RLWE}{\mathsf{RLWE}}
\newcommand{\GLWE}{\mathsf{GLWE}}
\newcommand{\GSW}{\mathsf{GSW}}
\newcommand{\RGSW}{\mathsf{RGSW}}
\newcommand{\GGSW}{\mathsf{GGSW}}

\newcommand{\Enc}{\mathsf{Enc}}

\newcommand{\Phase}{\mathsf{Phase}}
\newcommand{\Error}{\mathsf{Error}}

\newcommand{\FDFB}{\mathsf{FDFB}}

\newcommand{\Decomp}{\mathsf{Decomp}}
\newcommand{\ModSwitch}{\mathsf{ModSwitch}}

\newcommand{\KeySwitchSetup}{\mathsf{KeySwitchSetup}}
\newcommand{\KeySwitch}{\mathsf{KeySwitch}}

\newcommand{\KeyExtract}{\mathsf{KeyExt}}
\newcommand{\SampleExtract}{\mathsf{SampleExt}}
\newcommand{\BRKeyGen}{\mathsf{BRKeyGen}}
\newcommand{\BlindRotate}{\mathsf{BlindRotate}}
\newcommand{\Bootstrap}{\mathsf{Bootstrap}}

\newcommand{\PubMux}{\mathsf{PubMux}}

\newcommand{\inter}{\mathsf{inter}}

\newcommand{\BuildAcc}{\mathsf{BuildAcc}}

\newcommand{\ConstPoly}{\mathsf{SetupPolynomial}}
\newcommand{\BootsFunction}{\mathsf{BootsMap}}

\newcommand{\fresh}{\mathsf{fresh}}

\newcommand{\name}{\mathsf{mane}}

\renewcommand{\deg}{N}
 
\newcommand{\basis}{\mathsf{L}}
\newcommand{\bound}{\mathcal{B}}
\newcommand{\curr}{\mathsf{curr}}
\newcommand{\nextSym}{\mathsf{next}}
\newcommand{\BR}{\mathsf{BR}}

\newcommand{\TFHE}{\mathsf{TFHE}}

\newcommand{\boot}{\mathsf{boot}}

\newcommand{\preRotPoly}{\mathsf{rotP}}
\newcommand{\sgnRotPoly}{\mathsf{sgnP}}

\newcommand{\acc}{\mathsf{acc}}

\newcommand{\ct}{\mathsf{ct}}

\newcommand{\msg}{\mathsf{m}}

\newcommand{\keySwitchKey}{\mathsf{ksK}}
\newcommand{\brKey}{\mathsf{brK}} 

\newcommand{\packKey}{\mathsf{pK}}

\newcommand{\errDist}{\mathcal{X}}

\newcommand{\out}{\mathsf{out}}

\newcommand{\ReLU}{\mathsf{ReLU}}

\newcommand{\Softmax}{\mathsf{Softmax}}
\newcommand{\Dense}{\mathsf{Dense}}
\newcommand{\ConvTwoD}{\mathsf{Conv2D}}
\newcommand{\AvgPoolTwoD}{\mathsf{AvgPool2D}}

\theoremstyle{definition}
\newtheorem{definition}{\noindent\textbf{Definition}}
\newtheorem{lemma}{\textbf{Lemma}}
\theoremstyle{theorem}
\newtheorem{theorem}{\textbf{Theorem}}

\begin{document}

\title{FDFB: Full Domain Functional Bootstrapping Towards Practical Fully Homomorphic Encryption
}

\author{\IEEEauthorblockN{Kamil Kluczniak}
\IEEEauthorblockA{\textit{Stanford University,} \\
\textit{CISPA Helmholtz Center For Information Security}\\
kamil.kluczniak@\{stanford.edu,cispa.de\}}
\and
\IEEEauthorblockN{Leonard Schild}
\IEEEauthorblockA{\textit{CISPA Helmholtz Center For Information Security} \\
leonard.schild@cispa.de}
}

\maketitle

\begin{abstract}
 Computation on ciphertexts of all known fully homomorphic encryption (FHE) schemes induces some noise, which, if too large, will destroy the plaintext.
Therefore, the bootstrapping technique that re-encrypts a ciphertext and reduces the noise level remains the only known way of building FHE schemes for arbitrary unbounded computations.
The bootstrapping step is also the major efficiency bottleneck in current FHE schemes.
A promising direction towards improving concrete efficiency is to exploit the bootstrapping process to perform useful computation while reducing the noise at the same time.
 
We show a bootstrapping algorithm, which embeds a lookup table and evaluates arbitrary functions of the plaintext while reducing the noise. Depending on the choice of parameters, the resulting homomorphic encryption scheme may be either an exact FHE
or homomorphic encryption for approximate arithmetic.
Since we can evaluate arbitrary functions over the plaintext space, we can use the natural homomorphism of Regev encryption to compute affine functions without bootstrapping almost for free.
Consequently, our algorithms are particularly suitable for circuits with many additions and scalar multiplication gates. 
We achieve record speeds for such circuits. For example, in the exact FHE setting, we achieve a speedup of a factor of over 3000x over state-of-the-art methods. Effectively, we bring the evaluation time from weeks or days down to a few hours or minutes. Furthermore, we note that the speedup gets more significant with the size of the affine function.
  
We provide a tight error analysis and show several parameter sets for our bootstrapping. 
Finally, we implement our algorithm and provide extensive tests.  
We demonstrate our algorithms by evaluating different neural networks in several parameter and accuracy settings.

\end{abstract}


\section{Introduction}\label{sec:introduction}
 
A fully homomorphic encryption scheme provides a way to perform arbitrary computation on encrypted data.
The bootstrapping technique, first introduced by  Gentry \cite{STOC:Gentry09}, remains thus far the only technique to construct secure fully homomorphic encryption schemes. 
The reason is that in current homomorphic schemes evaluating encrypted data induces noise, which will eventually ``destroy'' the plaintext if too high.  
In practice bootstrapping also remains one of the 
major efficiency bottlenecks.

The most efficient bootstrapping algorithms to date are the FHEW-style bootstrapping developed by Ducas and Micciancio \cite{EC:DucMic15}, and the TFHE-style bootstrapping by Chillotti et al. \cite{AC:CGGI16,JC:CGGI20}.
At a high level, the bootstrapping procedure in these types of bootstrapping algorithms does two things. 
First, on input, a Regev ciphertext \cite{JAMC:Regev09} the bootstrapping algorithm outputs a ciphertext whose error is independent of the error of the bootstrapped ciphertext.
Second, it computes a negacyclic function $F$ on the input plaintext. That is $F$ must satisfy $F(x + \deg) = -F(x) \mod  Q$, where $x \in \ZZ_{2 \cdot \deg}$. 
This functionality, together with the linear homomorphism of Regev encryption \cite{JAMC:Regev09}, is enough for FHEW and TFHE to compute arbitrary binary gates. For example, given $\Enc(a \cdot 2\deg/3)$ and $\Enc(b \cdot 2\deg/3)$, where $a, b \in \{0, 1\}$, we compute the NAND gate by exploiting the negacyclicity. 
We set $F$ to output $-1$ for $x \in [0, \deg)$, and
from the negacyclicity property we have $F(x) = 1$ for  $x \in [\deg, 2 \cdot \deg)$. 
To compute the NAND gate we first compute
$\Enc(x) =  \Enc(a \cdot 2\deg/3) + \Enc(b  \cdot 2\deg/3)$ and then
 $\Enc(c) = \Enc(F(x))$. 
Note that only for $a = b = 1$,  we have $x = 4\deg/3 > \deg$ and $c = 1$. For all other valuations of $a$ and $b$ we have
$x < \deg$ and
 $c = -1$.
Now, we can exploit that $c \in \{-1, 1\}$ to choose one of two arbitrary values $y, z \in \NN$ by computing
$\frac{y \cdot (1 + c)}{2} - \frac{z \cdot (c - 1)}{2}$. When computing the NAND gate we choose $y = 0$ and $z = 2N/3$.
 Similarly, we can realize other boolean gates.
The crucial observation is that the bootstrapping algorithm relies on the negacyclicity property of the function $F$ to choose an outcome.
In particular, the outcome is a binary choice between two values. 
 
Unfortunately, the requirement on the function $F$ stays in the way to efficiently compute useful functions on non-binary plaintexts with only a single bootstrapping operation. 
For instance, we cannot compute $\frac{1}{1 + e^{-x}}$, $\tanh(x)$, $\max(0, x)$, univariate polynomials or other functions that are not negacyclic.
Hence, to evaluate such functions on encrypted data, we need to represent them as circuits, encrypt every input digit, and perform a relatively expensive bootstrapping per gate of the circuit.  
By solving the negacyclicity problems, we can hope for
significant efficiency improvements by computing all functions over a larger plaintext space with a single bootstrapping operation.  Additionally, we could leverage the natural and extremely efficient linear homomorphism of Regev encryption \cite{JAMC:Regev09},
to efficiently evaluate circuits with a large number of addition and scalar multiplication gates like neural networks.

 \subsection{Contribution}
 
Our main contribution is the design, detailed error analysis, implementation\footnote{Available at: https://github.com/cispa/full-domain-functional-bootstrap}, and extensive tests of a bootstrapping algorithm that solves the negacyclicity problem.
In particular, our bootstrapping algorithm can evaluate all functions over $\ZZ_{t}$ for some integer $t$, as it internally embeds a lookup table.
We will refer to our bootstrapping algorithm as ``full domain bootstrapping'' as opposed to a version of TFHE \cite{AC:CGGI16,JC:CGGI20,RSA:CarIzaMol19}, that can evaluate all function only over half of the domain $[0,\lfloor t/2 \rfloor)$.
We stress that the difference goes far beyond the size of the plaintext space. 
As we will explain in more detail in Section~\ref{sec:techniques},
TFHE for non-negacyclic functions cannot exploit the natural linear homomorphism of Regev encryption, whereas our $\FDFB$ can.
The ability to compute affine functions ``almost for free'' gives our $\FDFB$ a tremendous advantage in evaluation time for arithmetic circuits with a large number of addition and scalar multiplication gates.
We show several parameter sets targeting different bit-precisions and security levels.
In particular, we show parameters that allow our $\FDFB$ to bootstrap and compute any function $f:\ZZ_t \mapsto \ZZ_t$ with very low error probability, where $t$ is a 6, 7 or 8-bit integer, but we note that we bootstrap larger plaintexts when choosing a larger modulus $Q$ and degree $\deg$ of the ring $\R_Q$. We show that we can take the message modulus higher for approximate arithmetic—for instance, 10 or 11 bits.
We  also show how to leverage the Chinese remainder theorem to extend the plaintext space to $\ZZ_{t}$ for $t$ being a large (e.g., 32-bit) composite integer. In short, we use the fact that the ring $\ZZ_t$ for $t = \prod_{i=1}^n t_i$ where the $t_i$'s
are  pairwise co-prime
is isomorphic with the product ring $\ZZ_{t_1} \times \dots \times \ZZ_{t_n}$.
 
\textit{Performance.}
We implemented our bootstrapping algorithm using the Palisade library \cite{PALISADE}.
We exemplify the performance of our $\FDFB$ by evaluating neural networks.
Roughly speaking, a neural network is a circuit which gates, called neurons, parameterized by weights $w_0, w_1, \dots, w_n \in \NN$, take as input wires $x_1, \dots, x_n \in \NN$ and output  $f(w_0 + \sum_{i=1}^{n} x_i \cdot w_i)$.
The function $f$ is typically a non-linear function called the activation function.
When homomorphically evaluating an encrypted query to the neural network, we compute the linear combination $w_0 + \sum_{i=1}^{n} x_i \cdot w_i$ by leveraging  the natural linear homomorphism of Regev ciphertexts. 
Finally, we use our $\FDFB$ to compute the activation function $f$. 
Since our bootstrapping reduces the noise of the ciphertexts along with computing $f$, the time to evaluate the entire network is linear in the number of neurons.  
In particular, the evaluation time and the parameters for our FHE do not depend on the depth of the neural network.
A critical property of FHE schemes is exploiting the parallelism of the circuits we wish to evaluate.
On top of that, our functional bootstrapping algorithm is nicely parallelisable.

We compare our FHE scheme with the state-of-the-art lookup table methods \cite{AC:CGGI16,JC:CGGI20,AC:CGGI17}. 
To the best of our knowledge, the lookup methods are the fastest FHE able to perform exact homomorphic computation correctly.
To give the best quality of comparison, we rewrite the lookup table algorithms from \cite{AC:CGGI17} in PALISADE \cite{PALISADE}.
Furthermore, we choose new parameters for \cite{AC:CGGI16,JC:CGGI20,AC:CGGI17} that target the same security levels  as our bootstrapping algorithm.
We report that evaluating neurons with 784 weights is over 3000 to 600 times faster than the lookup table methods, depending on the parameter setting.
In particular, our homomorphic evaluation can be accomplished in minutes, whereas computing via the lookup tables requires weeks to accomplish\footnote{The time to evaluate the neural network was estimated for the lookup table methods based on the times to evaluate a single neuron.}.


\textit{Relation to BFV/BGV \cite{FOCS:BraVai11,C:Brakerski12,EPRINT:FanVer12,ITCS:BraGenVai12} type schemes.} 
In this setting, we consider FHE with a very low probability of having erroneous outcome of the homomorphic computation. We automatically satisfy CPA+-security \cite{EPRINT:LiMic20} in contrast to homomorphic encryption for approximate arithmetic. 
Current bootstrapping algorithms \cite{EC:HalSho15,EC:CheHan18,JC:HalSho21} for BGV/BFV type systems like HeLib  \cite{HElib,C:HalSho14,EC:HalSho15,EC:Albrecht17} or \cite{SEAL} reduce the error of the ciphertext and raise the modulus, without performing any additional computation. 
To compute in BFV/BGV, we represent the program as an arithmetic circuit.
This may be sometimes problematic. For example, \textit{CryptoNets} \cite{cryptonets}, which use the SEAL library to homomorphically evaluate
a neural network, cannot easily compute popular activation functions. Therefore, \textit{CryptoNets} uses $x^2$ as an activation function. 
Furthermore, we need to discretize the trained network by multiplying the weights and biases by a parameter $\delta$ and rounding.
Since there is no homomorphic division algorithm, homomorphic computation on discretised data will increase the parameter $\delta$. 
For example\footnote{We ignore the rounding for simplicity} when computing  $\sum_{i=1}^m x_i' \cdot w_i'$,
where $x_i = \delta \cdot x_i$ and $w_i' = \delta \cdot w_i$ and the $x_i$'s and $w_i$'s are the original floating point inputs and weights, we end up with $\delta^2 \cdot \sum_{i=1}^m x_i \cdot w_i$.
Unfortunately, we cannot rescale the result since there is no natural division algorithm in BFV/BGV.
Consequently, \textit{CryptoNets} has to choose enormous parameters that grow exponentially with the depth of the neural network.
Our functional bootstrapping algorithm resolves these issues, as we can compute any activation function.
Along the way, we can compute the rescaling, i.e., division by $\delta$ and rounding, together with the activation function within a single bootstrapping operation. This allows us to compute neural networks of unbounded depth without increasing the parameters.
On the other hand, we preserve the extremely fast linear operations from BFV/BGV.

\textit{Relation to CKKS/HEAAN \cite{AC:CKKS17} type schemes.}
As hinted in the abstract, we can run our bootstrapping algorithm in an approximate mode.
In particular, the probability that the error distorts
the least significant bits of the message to be bootstrapped grows along with
the message space.
 For example, we choose the plaintext space to be 11-bit, instead of 7-bit as in the ``exact'' computation setting.
However, when using such parameters, we cannot claim that the homomorphic encryption is a ``exact'' FHE anymore.
Consequently, as all approximate schemes we cannot formally claim CPA+-security \cite{EPRINT:LiMic20}, but
we can apply the same countermeasures against the Li-Micciancio attack.
Nevertheless, this setting allows to use larger message modulus with the same efficiency, and is especially
useful when approximate computation is sufficient (e.g., neural networks).
Note that without bootstrapping algorithms like ours, current approximate homomorphic encryptions need to approximate a function via a polynomial, discretize its coefficients and evaluate a circuit as deep as the degree of that polynomial.  

For the approximate homomorphic computation setting, Lu et al. \cite{SP:LHHMQ21} proposed to use FHEW \cite{EC:DucMic15} to compute negacyclic functions only\footnote{Actually, Lu et al. \cite{SP:LHHMQ21} claims to compute any function, but we noticed a serious flaw in their analysis. We address the issue in the next subsection.}. 
When using our algorithm, we can compute all functions over the plaintext space without posing any restriction on the plaintexts. 

\subsection{Deeper Dive into the Problem and our Solution}\label{sec:techniques}

 We first give a high-level overview of FHEW and TFHE bootstrapping to showcase the problem for using it to compute arbitrary functions. 
 In this section, we keep the exposition rather informal and omit numerous crucial details to highlight the essential ideas underlying our constructions.
 
\textit{Regev Encryption.} 
 First, let us recall the learning with errors based encryption due to Regev \cite{JAMC:Regev09}.
 In the symmetric key setting the encryption algorithm chooses a vector $\vec{a} \in \ZZ_{q}^{n}$ and a secret key $\vec{s} \in \ZZ_{q}^{n}$, and computes an encryption of a message $\msg \in \ZZ_{t}$ as $[b, \vec{a}^{\top}]^{\top} \in \ZZ_q^{(n+1) \times 1}$, with $b = \vec{a}^{\top} \cdot \vec{s} +  \tilde{\msg} + e \mod q$, where $\tilde{\msg} = \frac{q}{t}  \cdot \msg$ and $e < \frac{q}{2 \cdot t}$ is a ``small'' error. 
We assume that $t|q$ for simplicity, but other settings are possible as well. 
 To decrypt, we compute 
 $\bigg\lfloor\frac{t}{q} \cdot \big([b, \vec{a}^{\top}]^{\top} \cdot [1, -\vec{s}^{\top}]\big)\bigg\rceil 
 = \big\lfloor\frac{t}{q} \cdot(\frac{q}{t} \cdot  \msg + e)\big\rceil = \msg$.
The ring version of the encryption scheme is constructed over the cyclotomic ring $\R_Q$ defined by $\R = \ZZ[X]/(X^{\deg} + 1)$ and $\R_Q = \R/Q\R$.
Similarly as for the integer version we choose $\fr{a} \in \R_Q$ and a secret key $\fr{s} \in \R_Q$, and encrypt a message $\msg \in \R_{t}$ as $[\fr{b}, \fr{a}]$ with $\fr{b} = \fr{a} \cdot \fr{s} + \frac{q}{t}  \cdot \msg + \fr{e}$, where $e \in \R_Q$ is a ``small'' error.
 
\textit{FHEW/TFHE Bootstrapping.}
Now let us proceed to the ideas underlying FHEW \cite{EC:DucMic15} and TFHE   \cite{AC:CGGI16,JC:CGGI20} to bootstrap the LWE ciphertext described above. 
Both algorithms leverage the structure of the ring $\R_Q$. Recall that $\R_Q$ includes polynomials from $\ZZ[X]/(X^{\deg} + 1)$
that have coefficients in $\ZZ_Q$. Importantly, in $\R_Q$ the roots of unity form a multiplicative cyclic group $\mathbb{G} = [1, X, \dots, X^{\deg-1}, -1, -X, \dots, -X^{\deg-1}]$ of order $2 \cdot \deg$\footnote{This is trivial to verify by checking that $X^{\deg} \mod (X^{\deg}+1) = -1$.}.
Assume that the LWE modulus is $q = 2 \cdot \deg$. 
The concept is to setup a homomorphic accumulator $\acc$ to be an RLWE encryption which initially encrypts the message $\preRotPoly \cdot X^{b} \in \R_Q$, where 
$\preRotPoly = \frac{Q}{t} \cdot \big(1 - \sum_{i=2}^{\deg} X^{i-1}\big) \in R_Q$. 
Then we multiply $\acc$ with encryptions of $X^{- \vec{a}[i] \cdot \vec{s}[i]} \in \R_Q$.
 So that after $n$ iterations the message of the accumulator is set to 
\begin{align*}
 \preRotPoly \cdot X^{b - \sum_{i=1}^n \vec{a}[i] \cdot \vec{s}[i]} &=    \preRotPoly \cdot X^{k \cdot q + \tilde{\msg} + e} \\ 
&=   \preRotPoly \cdot X^{\tilde{\msg} + e \mod 2\cdot \deg} \in \R_Q.
\end{align*}
Note that we sent the coefficients of $\preRotPoly$ such that the constant coefficient of the resulting polynomial
is $\frac{Q}{t}$ if  $\tilde{\msg} + e \in [0, \deg)$, and $- \frac{Q}{t}$ if $\tilde{\msg} + e \in [\deg, 2\deg)$.  
 We can extract an LWE encryption of the constant term from the rotated accumulator and obtain an LWE ciphertext of the sign\footnote{Clearly, if the most significant bit of $\msg$ is set ($\msg$ is a negative number) and $e \le q/2t$, then $\tilde{\msg} + e = q/t \cdot \msg + e \in  [\deg, 2\deg)$. Otherwise, $\tilde{\msg}  + e \in [0, \deg)$.} of the message without the error term $e$.
 
We can extend the method to compute other functions simply by setting the coefficient of $\preRotPoly$ to have the desired value
after multiplying it by $X^{\tilde{\msg} + e}$ in $\R_Q$. 
To be precise we want to compute $f:\ZZ_t \mapsto \ZZ_t$, and for our reasoning we use $F:2 \cdot \deg \mapsto Q$ such that
$F(x) = \frac{Q}{t} \cdot f(\lfloor \frac{t}{2\cdot \deg} \cdot x \rceil)$.
Since we work over the multiplicative group $\mathbb{G}$ of order $2\cdot \deg$, but $\preRotPoly$ can only have $\deg$
coefficients, we can only compute functions $F$ such that $F(x + \deg \mod 2 \cdot \deg) = -F(x) \mod Q$. 
In other words, multiplying any $\preRotPoly \in R_Q$ by $X^{\deg}$ rotates $\preRotPoly$ by a full cycle negating all its coefficients.
In particular, it flips the sign of the constant coefficient. 
Note that this behavior also affects the function $f$ that we want to compute.
 
As we mentioned earlier, Lu et al. \cite{SP:LHHMQ21} claim (see \cite[Fig.~2,Thm~1]{SP:LHHMQ21}) to compute any function with FHEW \cite{EC:DucMic15}.
Unfortunately, they mistakenly assume the roots of unity in $\ZZ/(X^{\deg} +1)$ form the group $[1, X, \dots, X^{\deg}]$ of order $\deg$,
instead of $\mathbb{G}$ of order $2 \deg$. This issue invalidates \cite[Thm~1]{SP:LHHMQ21} and the correctness of their main contribution, as FHEW is in fact only capable to compute negacyclic functions.

%
 
 \textit{Problems when Computing Arbitrary Functions.}
To compute arbitrary functions $f$ a straightforward solution is assume  
 that the phase $\tilde{\msg} + e$ is always in $[0, \deg) = [0, q/2)$ which is the case when $\msg \in [0,t/2)$.
 In particular, $\msg$ must have its sign fixed.
 The immediate disadvantage is that we cannot safely compute affine functions on LWE ciphertexts without invoking the bootstrapping algorithm. 
As an example let us consider $\tilde{\msg} = \frac{q}{t}(\msg_1 - \msg_2)$ and let us ignore for simplicity the error terms.
Clearly when $\msg_2 > \msg_1$ we have $\msg_1 - \msg_2 \mod t \in [t/2,t)$ and $\tilde{\msg} \in [q/2,q)$.
In this case, the output of the bootstrapping will be $-F(\tilde{\msg} - q/2)$ assuming we set the rotation polynomial to correctly compute $F$ on the interval $[0, q/2)$.
 To summarize, to use the natural functional bootstrapping correctly, we need to use it
 as  Carpov et al. \cite{RSA:CarIzaMol19} to compute a lookup table on an array of binary plaintexts that are composed to an integer within the interval $[0,t/2)$. 
Unfortunately, we cannot correctly compute the extremely efficient affine functions on LWE ciphertexts.

\textit{\textbf{Our Solution:} Full Domain Functional Bootstrapping.}
Our main contribution is the design of a bootstrapping algorithm that can compute all functions on an encrypted plaintext
instead of only negacyclic functions.
Our first observation is that, as shown in Bourse et al. \cite{C:BMMP18}, we can compute the sign of the message.
In particular, given that the message $\msg$ is encoded as $\tilde{\msg} = \frac{2 \cdot \deg}{t} \cdot \msg$, where $t$ is the plaintext modulus, we have that
if $\tilde{\msg} \in [0, \deg)$, then $\sgn(\msg) = 1$, and $\tilde{\msg} \in [\deg, 2 \cdot \deg)$, then
$\sgn(\msg) = -1$. Note that this function does not return exactly the sign function since it returns $1$ for $\msg=0$, but it is enough for our purpose. 
The idea is to make two bootstrapping operations: the first computes the sign; the second computes the function $F(x)$.
Then we multiply the sign and the result of the function.
Note that if $x \in [\deg, 2 \cdot \deg)$, then $-F(x) \cdot \sgn(x) = F(x)$.

There are two problems with this solution. 
First of all, we cannot simply multiply two LWE ciphertexts in the leveled mode. 
While there are methods to do so \cite{C:Brakerski12,FOCS:BraVai11,ITCS:BraGenVai12,EPRINT:FanVer12}, in these methods, the noise growth is quadratic and dependent on the size of the plaintexts. 
Moreover, one of those plaintexts may be large having a significant impact on the parameters. 
The second problem is that this way, the function $F$ must satisfy $F(x) = F(x + \deg)$. 
Thus our primary goal is not satisfied.

Let us first describe how to deal with the second problem.
In short, instead of evaluating the sign, we will compute a bit $b$ that indicates the interval in which that plaintext lies. 
This bit will help us choose the correct output of the evaluated function. That is, we split the function into $F_0$ that computes $F$ correctly for plaintexts in $[0, \deg)$, and $F_1$ that computes correctly for the remaining plaintexts.
To deal with the first problem, we resign of multiplying two ciphertexts at all.
Note that a hypothetical multiplication algorithm could now bootstrap the bit $b$ using a bootstrapping that returns a GSW tuple, and in a leveled mode we could then compute $\text{GSW}(1-b) \cdot \LWE(F_0(x)) + \text{GSW}(b) \cdot \LWE(F_1(x))$.
From the properties of GSW cryptosystem \cite{C:GenSahWat13}, the resulting ciphertexts' noise is at most an additive function of the noises of both ciphertexts.
Furthermore, there are implementations for bootstrapping that return GSW ciphertexts \cite{EC:DucMic15,AC:CGGI17}.
We, however, slightly depart from the idea and give a more efficient solution with better error management. 
Instead of multiplying two ciphertexts, we leverage the fact that the functions $F_0$ and $F_1$ are publicly known, which means that the rotation polynomials for these functions are known as well.
Now, based on LWE ciphertexts of a certain form, we will build a new accumulator that will encode one of the given rotation polynomials.
Notably, our accumulator builder uses a new version of the homomorphic CMux gate \cite{AC:CGGI16,JC:CGGI20} that may be of separate interest.
 Finally, we will bootstrap using the given accumulator and we obtain $F(x) = F_{b}(x)$ as desired.
We note that this method is roughly twice as fast as creating a GSW and multiplying it with LWE ciphertexts.
 
Now we can compute any function $F:2\cdot \deg \mapsto Q$, and what follows  $f:\ZZ_t \mapsto \ZZ_t$\footnote{We note that we can easily generalize the method to also compute functions $f:\ZZ_t \mapsto \ZZ_p$ where $t \neq p$. It remains to use a different message scaling in the rotation polynomials.}. 
Furthermore, we can compute affine functions over $\ZZ_t$ without bootstrapping which costs several microseconds per operation.

\subsection{Related Work}\label{sec:related-work}

Since Gentry's introduction of the bootstrapping technique  \cite{STOC:Gentry09}, the design of fully homomorphic encryption schemes has received extensive study, e.g., \cite{FOCS:BraVai11,C:Brakerski12,EPRINT:FanVer12,ITCS:BraGenVai12,C:AlpPei13,C:GenSahWat13,EC:HalSho15,EC:CheHan18,JC:HalSho21}.
Brakerski and Vaikuntanathan \cite{ITCS:BraVai14} achieved a breakthrough and gave a bootstrapping method that, exploiting the GSW cryptosystem \cite{C:GenSahWat13} by Gentry, Sahai, and Waters, incurs only a polynomial error growth.
These techniques require representing the decryption circuit as a branching program.  
Alperin-Sheriff and Peikert \cite{C:AlpPei14} showed a bootstrapping algorithm, where they represent the decryption circuit as an arithmetic circuit, and we do not  rely on Barrington's theorem \cite{BARRINGTON}. 
Their method exploits the GSW cryptosystem \cite{C:GenSahWat13} to perform matrix operations on LWE ciphertexts.
Hiromasa, Abe, and Okamoto \cite{PKC:HirAbeOka15} improved the method by designing a version of GSW that natively encrypts matrices. Recently, Genise et al. \cite{AC:GGHLM19} showed an encryption scheme that further improves the efficiency of matrix operations. Unfortunately, Lee and Wallet \cite{EPRINT:LeeWal20} showed that the scheme is broken for practical parameters.

 Building on the ideas of Alperin-Sheriff and Peikert \cite{C:AlpPei14}, Ducas and Miccancio \cite{EC:DucMic15} designed a practical bootstrapping algorithm that we call FHEW. FHEW exploits the ring structure of the RLWE ciphertexts and RGSW noise management for updating a so-called homomorphic accumulator. 
Chillotti et al. \cite{AC:CGGI16,JC:CGGI20} gave a variant, called TFHE, with numerous optimizations to the FHEW bootstrapping algorithm.

The FHEW and TFHE bootstrapping algorithms constitute state-of-the-art methods to perform practical bootstrapping.
 Further improvements mostly relied on incorporating packing techniques \cite{AC:CGGI17,ICALP:MicSor18},
and improved lookup tables evaluation \cite{AC:CGGI17,RSA:CarIzaMol19}.
%
An interesting direction is to exploit the construction of FHEW/TFHE and embed a lookup table directly into the bootstrapping algorithm.
To this end, Bonnoron, Ducas, and Fillinger \cite{AFRICACRYPT:BonDucFil18} showed a bootstrapping algorithm capable of computing over larger plaintexts, which, however, returns a single bit of output per bootstrapping operation and uses RLWE instantiated over cyclic rings instead of negacyclic. 
Carpov,  Izabach\'ene, and Mollimard \cite{RSA:CarIzaMol19} show that we can compute multiple
functions on the same input plaintext at the cost of a single TFHE bootstrapping. 
Notably, their method enjoys amortized time only for functions with a small domain like binary.
 Guimar\~aes, Borin, and Aranha \cite{CHES:GuiBorAra21} show several applications of \cite{RSA:CarIzaMol19}.
 As mentioned earlier, 
Lu et al. \cite{SP:LHHMQ21} use FHEW to compute a negacyclic functions on CKKS \cite{AC:CKKS17} ciphertexts.

\textit{Applications to Neural Network Inference.}
Bourse et al. \cite{C:BMMP18} applied TFHE  \cite{AC:CGGI16,JC:CGGI20} to compute the negacyclic $\sgn(x)$ function as activation for a neural network. 
 Izabach\'ene, Sirdey,  and  Zuber \cite{CANS:IzaSirZub19} evaluate another neural network (a Hopfield network) using the same method as Bourse et al. \cite{C:BMMP18}. 
 Finally, we note that Bourse et al. \cite{C:BMMP18} left as an open problem to compute general functions with a functional bootstrapping.
It is worth noting that oblivious neural network evaluation became a critical use case for fully homomorphic encryption and got much attention and publicity \cite{cryptonets,EPRINT:CWMMP17,CCS:LJLA17,CCS:JKLS18,USENIX:JuvVaiCha18,BMC:KSKLC18,BMC:KSKK20,EPRINT:BGGJ18,RSA:ABSV19,CCS:CDKS19,USENIX:RSCLLK19,PNAS:BlaGusPol20}. 
These works focus on applying existing homomorphic encryption libraries like SEAL, HeLib, and HEAAN, which we discussed earlier.


\textit{Concurrent Work.}
 Concurrently and independently, Chillotti et al. \cite{EPRINT:CLOT21} shows a concept for a full domain functional bootstrapping algorithm that at a very high level is similar to ours. We stress that the methods to achieve FDFB are different. 
 Chillotti et al. \cite{EPRINT:CLOT21} propose to use the BGV cryptosystem instead of GSW as we suggested for a hypothetical solution, to multiply and choose one of two bootstrapped values. 
  In contrast to our work, the authors do not show any correctness analysis, parameters, or implementation.
  Furthermore, we already implemented and tested neural network inference, while  Chillotti et al. \cite{EPRINT:CLOT21} admit having problems with implementing their algorithm on top of the TFHE library at the time of writing this article.
 Nevertheless, we note that due to the use of BGV style multiplication, the error after bootstrapping grows quadratically and is dependent on the size of the message. In contrast, our algorithm has a smaller error that is independent of the message. Recall that we want to bootstrap large messages. Hence it is not clear what impact this will have on concrete parameters and efficiency. 

\section{Preliminaries}\label{sec:preliminaries}

In this section, we recall the algorithms from the literature that we use for our bootstrapping algorithm.
Here we give only the interfaces to the algorithm and state their correctness lemmas.  
However, for reference, we give the specifications and the correctness analysis of all the algorithms in a unified notation in \append~\ref{sec:correctness-for-preliminaries}.


\textbf{Notation.}
We denote as $\BB$ the set $\{0, 1\}$, $\RR$ as the set of reals, $\ZZ$ the set of integers, and $\NN$ as the set of natural numbers. 
Furthermore we denote as $\R_{\deg, \infty}$ the ring of polynomials $\ZZ[X]/(X^{\deg} + 1)$ and as $\R_{\deg, q} = \R_{\deg, \infty}/q\ZZ$ the ring of polynomials with coefficients in $\ZZ_q$.  
We denote vectors with a bold lowercase letter, e.g., $\vec{v}$,
and matrices with uppercase letters $\ma{V}$. 
 We denote a $n$ dimensional column vector as $[f(.,i)]_{i=1}^{n}$, where $f(.,i)$ defines the $i$-th coordinate. For brevity, we will also denote as $[n]$ the vector $[i]_{i=1}^{n}$, and more generally $[n,m]_{i=n}^{m}$ the vector $[n, \dots, m]^{\top}$.
Finally, let $a = \sum_{i=1}^{\deg} a_i \cdot X^{i-1}$, be a polynomial with coefficients over $\RR$, then we denote $\coefs(a) = [a_i]_{i=1}^{\deg} \in \RR^{\deg \times 1}$ the vector of coefficients of the polynomial $a$. 
For a random variable $a \in \ZZ$ we denote as $\Var(a)$ the variance of $a$, as $\stddev(x)$ its standard deviation and as $\E(x)$ its expectation.
For $\fr{a} \in \R_{\deg, \infty}$, we define $\Var(\fr{a}) = [\Var(\coefs(\fr{a}))]_{i=1}^{\deg}$, $\stddev(\fr{a}) = [\stddev(\coefs(\fr{a}))]_{i=1}^{\deg}$ and $\E(\fr{a}) = [\E(\coefs(\fr{a}))]_{i=1}^{\deg}$.   
We denote $\norm{\vec{a}}_p = (\sum_{i=1}^p \abs{a_i}^p)^{1/p}$ the $p$-norm of a vector $\vec{a} \in \RR^n$,
where $\abs{.}$ denotes the absolute value.
For polynomials  
we compute the $p$-norm by taking its coefficient vector. 
By $\hamming(\vec{a})$ we denote the hamming weight of vector $\vec{a}$, i.e., the number of of non-zero coordinates of $\vec{a}$. 
We also define a special symbol $\Delta_{q, t} = \round{\frac{q}{t}}$ and a rounding function for an element $\Delta_{q, t} \cdot a \in \ZZ_q$,
 as $\round{a}_{t}^q = \round{\frac{t}{q} \cdot \Delta_{q, t} \cdot a}$. For ring elements, we take the rounding of the coefficient vector.

Throughout the paper we denote as  $q \in \NN$ and $Q \in \NN$ two moduli.
The parameter $n \in \NN$ always denotes the dimension of a LWE sample, that we define below.
For rings, we always use $\deg$ to denote the degree of $\R_{\deg, q}$ or $\R_{\deg, Q}$.
We define $\ell = \roundUp{\log_{\basis} q}$ for some decomposition basis $\basis \in \NN$.
We denote bounds on variances of random variables by $\bound \in \NN$.
Often we mark different decomposition bases $\basis_{\name}$ and the corresponding $\ell_{\name}$, or bounds $\bound_{\name}$ with
some subscript $\name$. 
Finally, in order not to repeat ourselves and not to overwhelm the reader, we limit ourselves
to define certain variables in the definitions of the algorithms, and we refrain to repeat them in the correctness lemmas.

%
%
%
%
%

\textbf{Learning With Errors.} 
We recall the learning with errors assumption by Regev \cite{STOC:Regev05}.
Our description is a generalized version due to Brakerski, Gentry, and Vaikuntanathan \cite{ITCS:BraGenVai12}.

\begin{definition}[Generalized Learning With Errors]
Let $\errDist_{\bound} = \errDist(\secpar)$ be a distribution over $\R_{\deg, q}$
such that  $\norm{\Var(\fr{e})}_{\infty} \leq \bound$
and $\E(\fr{e}) = \vec{0}$
 for all $\fr{e} \in \errDist_{\bound}$. 
For $\vec{a} \sample \R_{\deg, q}^{n \times 1}$, $\fr{e} \sample \errDist_{\bound}$
and $\vec{s} \in \R_{\deg, q}^{n \times 1}$, 
we define a Generalized Learning With Errors (GLWE) sample of a message $\msg \in \R_{\deg, q}$ with respect to $\vec{s}$, as
\begin{align*}
\GLWE_{\bound, n, \deg, q}(\vec{s}, \msg) = 
\begin{bmatrix}
\fr{b}  = \vec{a}^{\top} \cdot \vec{s} + \fr{e}\\
\vec{a}^{\top}
\end{bmatrix} + 
\begin{bmatrix}
\msg \\
\vec{0}
\end{bmatrix}  \in \R_{\deg, q}^{(n+1) \times 1}.
\end{align*}
The $\GLWE_{\bound, n, \deg, q}$ problem is to distinguish the following two distributions:
In the first distribution, one samples $(\fr{b}_i, \vec{a}_i^\top)^\top$ uniformly from $\R_{\deg, q}^{(n + 1) \times 1}$.
In the second distribution one first draws 
$\vec{s} \sample \R_{\deg, q}^{n \times 1}$ 
and then sample $\GLWE_{\bound, n, \deg, q}(\vec{s}, \msg) = [\fr{b}_i, \vec{a}_i^{\top}]^{\top} \in \R_{\deg, q}^{(n+1) \times 1}$. 
The $\GLWE_{\bound, n, \deg, q}$ assumption is that the $\GLWE_{\bound, n, \deg, q}$ problem is hard.
\end{definition}

In the definition, we choose the secret key uniformly from $\R_{\deg, q}$.
We note that often secret keys may be chosen from different distributions.
When describing our scheme, we will emphasize this
when necessary.

We denote an Learning With Errors (LWE) sample as $\LWE_{\bound, n, q}(\vec{s}, \msg) = \GLWE_{\bound, n, 1, q}$, which is a special case of a GLWE sample with $N=1$.
Similarly we will denote an Ring-Learning with Errors (RLWE) 
sample as  $\RLWE_{\bound, \deg, q}(\vec{s}, \msg) = \GLWE_{\bound, 1, \deg, q}$ which is the special case of an GLWE sample with $n=1$.
Sometimes we will use the notation $\vec{c} \in \GLWE_{\bound, n, \deg, q}(\vec{s}, \msg)$ (resp. $\LWE$ and $\RLWE$) to indicate
that a vector $\vec{c}$ is a GLWE (resp. LWE and RLWE) sample of the corresponding parameters and inputs.
Sometimes we will leave the inputs unspecified and substitute them with "." when it is not necessary to refer to them within the scope of a function.


\begin{definition}[Phase and Error for GLWE samples]
We define the phase of a sample $\vec{c} = \GLWE_{\bound, n, \deg, q}(\vec{s}, \msg)$,
as $\Phase(\vec{c}) = [1,  - \vec{s}^{\top}]  \cdot \vec{c}$.  
We define  the error of a GLWE sample $\Error(\vec{c}) =  \Phase(\vec{c}) - \msg$. 
\end{definition}

\begin{lemma}[Linear Homomorphism of GLWE samples]\label{lemma:linear-homo-glwe}
Let $\vec{c} = \GLWE_{\bound_{\vec{c}}, n, \deg, q}(\vec{s}, \msg_{\vec{c}})$
and $\vec{d} = \GLWE_{\bound_{\vec{d}}, n, \deg, q}(\vec{s}, \msg_{\vec{d}})$. 
If $\vec{c}_{\out} \exec  \vec{c} + \vec{d}$, then
$\vec{c}_{\out} \in \GLWE_{\bound, n, \deg, q}(\vec{s}, \msg)$,
where $\msg = \msg_{\vec{c}} + \msg_{\vec{d}}$,   and 
$\bound \leq \bound_{\vec{c}} + \bound_{\vec{d}}$.  
Furthermore, let $\fr{d} \in \R_{\deg, \basis}$ where $\basis \in \NN$.
If $\vec{c}_{\out} \exec \vec{c} \cdot \fr{d}$, then $\vec{c}_{\out} \in \GLWE_{\bound, n, \deg, q}(\vec{s}, \msg_{\vec{c}} \cdot \fr{d})$,
where  
\begin{itemize}
\item $\bound \leq \bound_{\fr{d}}^{2} \cdot \bound_{\vec{c}}$, when $\fr{d}$ is a constant monomial with $\norm{\fr{d}}_{\infty} \leq \bound_{\fr{d}}$,

\item $\bound \leq \frac{1}{3} \cdot \basis^2 \cdot \bound_{\vec{c}}$, when $\fr{d}$ is a monomial with its non-zero coefficient distributed uniformly in $[0, \basis-1]$,

\item $\bound \leq \deg \cdot \bound_{\fr{d}}^{2} \cdot \bound_{\vec{c}}$, when $\fr{d}$ is a constant polynomial with $\norm{\fr{d}}_{\infty} \leq \bound_{\fr{d}}$, and 

\item $\bound \leq \frac{1}{3} \cdot \deg  \cdot \basis^2 \cdot \bound_{\vec{c}}$, when $\fr{d}$ is a polynomial with its  coefficients distributed uniformly in $[0, \basis-1]$.
\end{itemize} 
\end{lemma}


\textbf{Gentry-Sahai-Waters Encryption.} 
We recall the cryptosystem from Gentry, Sahai and Waters \cite{C:GenSahWat13}.
 
\begin{definition}[Gadget Matrix]\label{def:gadgetdecomp}
We call the column vector $\vec{g}_{\ell, \basis} =  [\basis^{i-1}]_{i=1}^{\ell}  \in \NN^{\ell \times 1}$ 
powers-of-$\basis$ vector.
For some $k \in \NN$ we define the gadget matrix $\ma{G}$ as 
$\ma{G}_{\ell, \basis, k} = \ma{I}_{k} \tensor \vec{g}_{\ell, \basis}  \in \NN^{k \ell \times k}$,
where $\tensor$ denotes the Kronecker product.
\end{definition}

\begin{definition}[Decomposition]
We define the $\Decomp_{\basis, q}$ function to take as input
an element $\fr{a} \in \R_{\deg, q}$ and return a row vector $\vec{a} = [\fr{a}_1, \dots, \fr{a}_{\ell}] \in \R_{\basis}^{1 \times \ell}$
such that $\fr{a} = \vec{a} \cdot \vec{g}_{\ell, \basis} =  \sum_{i=1}^{\ell} \fr{a}_i \cdot \basis^{i-1}$.
Furthermore we generalize the function to matrices where $\Decomp_{\basis, q}$ is applied
to every entry of the input matrix.
Specifically, on input 
a matrix $\ma{M} \in \R_{\deg, q}^{m \times k}$, $\Decomp_{\basis, q}$ outputs
a matrix $\ma{D}  \in \R_{\basis}^{m \times k \ell}$.
\end{definition}
 
\begin{definition}[Generalized-GSW Sample]
Let $\vec{s} \in \R_{\deg, q}^n$ and $\msg \in \R_{\deg, q}$. 
We define Generalized-GSW (GGSW) samples as 
$\GGSW_{\bound, n, \deg, q, \basis}(\vec{s}, \msg) = \ma{A} + \msg \cdot \ma{G}_{\ell, \basis, (n+1)} \in \R_{\deg, q}^{(n+1) \ell \times (n+1)}$,
where
$\ma{A}[i,\ast] = \GLWE_{\bound, n, \deg, q}(\vec{s}, 0)^{\top}$ for all $i \in [(n+1) \ell]$.
In other words, all rows of $\ma{A}$ consist of (transposed) GLWE samples of zero.
\end{definition}

Similarly, as with LWE and RLWE, we denote an GSW sample as $\GSW_{\bound, n, q, \basis}(\vec{s}, \msg) = \GGSW_{\bound, n, 1, q, \basis}(\vec{s}, \msg)$,
which is a special case of a GGSW sample with $\deg = 1$.
Similarly we will denote an Ring-GSW (RGSW)
sample as  $\RGSW_{\bound, \deg, q, \basis}(\vec{s}, \msg) = \GGSW_{\bound, 1, \deg, q, \basis}(\vec{s}, \msg)$,
which is the special case of an GGSW sample with $n=1$. 
Will use the notation $\ma{C} \in \GGSW_{\bound, n, \deg, q, \basis}(\vec{s}, \msg)$ (resp. $\GSW$ and $\RGSW$) to indicate that a matrix $\ma{C}$ is a GGSW (resp. GSW and RGSW) sample with the corresponding parameters and inputs.

\begin{definition}[Phase and Error for GGSW samples]
We define the phase of a sample $\ma{C} = \GGSW_{\bound, n, \deg, q, \basis}(\vec{s}, \msg) = \ma{A} + \msg \cdot \ma{G}_{\ell, \basis, (n+1)}$ 
as $\Phase(\ma{C}) = \ma{C} \cdot [1,  - \vec{s}^{\top}]^{\top}$. 
We define  the error of a GLWE sample  as
$\Error(\ma{C}) =  \Phase(\ma{C}) - \msg \cdot \ma{G}_{\ell, \basis, (n+1)} \cdot [1,  - \vec{s}^{\top}]^{\top} \in \R_{\deg, q}^{(n+1)\ell \times 1}$. 
\end{definition}

\begin{lemma}[Linear Homomorphism of GGSW samples]
\label{lemma:linear-homo-ggsw}
Let $\ma{C} = \GGSW_{\bound_{\ma{C}}, n, \deg, q, \basis}(\vec{s}, \msg_{\ma{C}})$
and $\ma{D} = \GGSW_{\bound_{\ma{D}}, n, \deg, q, \basis}(\vec{s}, \msg_{\ma{D}})$. 
If $\ma{C}_{\out} \exec \ma{C} + \ma{D}$, then
$\ma{C}_{\out} \in \GGSW_{\bound, n, \deg, q, \basis}(\vec{s}, \msg)$,
where $\msg = \msg_{\ma{C}} + \msg_{\ma{D}}$, and
$\bound \leq \bound_{\ma{C}} + \bound_{\ma{D}}$.
Furthermore, let $\fr{d} \in \R_{\basis_{\fr{d}}}$, where $\basis_{\fr{d}} \in \NN$.
If $\ma{C}_{\out} \exec \ma{C} \cdot \fr{d}$, then $\ma{C}_{\out} \in \GGSW_{\bound, n, \deg, q, \basis}(\vec{s}, \msg_{\ma{C}} \cdot \fr{d})$ 
where 
\begin{itemize}
\item $\bound \leq \bound_{\fr{d}}^{2} \cdot \bound_{\vec{c}}$, when $\fr{d}$ is a constant monomial with $\norm{\fr{d}}_{\infty} \leq \bound_{\fr{d}}$,

\item $\bound \leq \frac{1}{3} \cdot \basis^2 \cdot \bound_{\vec{c}}$, when $\fr{d}$ is a monomial with its non-zero coefficient distributed uniformly in $[0, \basis-1]$,

\item $\bound \leq \deg \cdot \bound_{\fr{d}}^{2} \cdot \bound_{\vec{c}}$, when $\fr{d}$ is a constant polynomial with $\norm{\fr{d}}_{\infty} \leq \bound_{\fr{d}}$, and 

\item $\bound \leq \frac{1}{3} \cdot \deg \cdot \basis^2 \cdot \bound_{\vec{c}}$, when $\fr{d}$ is a polynomial with its  coefficients distributed uniformly in $[0, \basis-1]$.
\end{itemize}  
\end{lemma}

The external product multiplies an RGSW sample with an RLWE sample resulting in an RLWE sample of the product of their messages.

\begin{definition}[External Product]
The external product $\extProd$ given as input $\ma{C} \in \GGSW_{\bound_{\ma{C}}, n, \deg, q, \basis}(\vec{s}, \msg_{\vec{C}})$ and $\vec{d} \in \GLWE_{\bound_{\vec{d}}, n, \deg, q}(\vec{s}, \msg_{\vec{d}})$, 
outputs $\extProd(\ma{C}, \vec{d}) = \Decomp_{\basis, q}(\vec{d}^{\top}) \cdot \ma{C}$.
\end{definition}

\begin{lemma}[Correctness of the External Product]
\label{lemma:correctness-external-product}
Let  $\ma{C} = \GGSW_{\bound_{\ma{C}}, n, \deg, q, \basis}(\vec{s}, \msg_{\ma{C}})$,
where $\msg_{\ma{C}}$ consists of a single coefficient,
and 
$\vec{d} = \GLWE_{\bound_{\vec{d}}, n, \deg, q}(\vec{s}, \msg_{\vec{d}})$.
If $\vec{c} \exec \extProd(\ma{C}, \vec{d})$, then
$\vec{c}^{\top} \in \GLWE_{\bound, n, \deg, q}(\vec{s}, \msg)$, where $\msg = \msg_{\ma{C}} \cdot \msg_{\vec{d}}$ 
and 
\begin{itemize}
\item $\bound \leq \frac{1}{3}  \cdot \deg  \cdot (n + 1) \cdot \ell \cdot \basis^2 \cdot \bound_{\ma{C}} + \deg \cdot \bound_{\msg_{\ma{C}}}^2 \cdot \bound_{\vec{d}}$ in general, and

\item $\bound \leq \frac{1}{3}  \cdot \deg \cdot (n + 1) \cdot \ell  \cdot \basis^2 \cdot \bound_{\ma{C}} + \bound_{\msg_{\ma{C}}}^2 \cdot \bound_{\vec{d}}$ when $\msg_{\ma{C}}$ is a monomial.
\end{itemize}  
with $\bound_{\msg_{\vec{C}}} \leq \norm{\msg_{\vec{C}}}_{\infty}$.
\end{lemma}

\textbf{CMux Gate.} 
The CMux gate was first introduced by Chillotti et al. \cite{AC:CGGI16,JC:CGGI20}.
Informally, it takes as input a control GGSW sample and two GLWE samples.
The gate outputs one of the input GLWE samples depending on the bit encrypted in the GGSW sample.

\begin{definition}[CMux Gate]\label{def:cmux}
The $\CMux$ function is defined as follows.
\begin{itemize}
\item $\CMux(\ma{C}, \vec{g}, \vec{h})$: On input
	$\ma{C} \in \GGSW_{\bound_{\ma{C}}, n, \deg, q, \basis}(\vec{s}, \msg_{\ma{C}})$, 
where $\msg_{\ma{C}} \in \BB$,  $\vec{g} \in \GLWE_{\bound_{\vec{g}}, n, \deg, q}(\vec{s}, .)$ and 
$\vec{h} \in \GLWE_{\bound_{\vec{h}}, n, \deg, q}(\vec{s}, .)$, 
the function returns $\vec{c}_{\out} \in \GLWE_{\bound, n, \deg, q}(\vec{s},, .)$. 
\end{itemize}
\end{definition}

\begin{lemma}[Correctness of the $\CMux$ gate]\label{lemma:correctness-of-cmux}
Let  $\ma{C} \in \GGSW_{\bound_{\ma{C}}, n, \deg, q, \basis}(\vec{s}, \msg_{\ma{C}})$,
where $\msg_{\ma{C}} \in \BB$ . 
Let $\vec{g} \in \GLWE_{\bound_{\vec{g}}, n, \deg, q}(\vec{s}, \msg_{\vec{g}})$
and $\vec{h} \in \GLWE_{\bound_{\vec{h}}, n, \deg, q}(\vec{s}, \msg_{\vec{h}})$. 

If $\vec{c}_{\out} \exec \CMux(\ma{C}, \vec{g}, \vec{h})$, then
$\vec{c}_{\out} \in \GLWE_{\bound, n, \deg, q}(\vec{s}, \msg_{\out})$,
where $\msg_{\out} = \msg_{\vec{h}}$ for $\msg_{\ma{C}} = 0$ and $\msg_{\out} = \msg_{\vec{g}}$ for $\msg_{\ma{C}} = 1$,
and  
$\bound \leq \frac{1}{3}  \cdot (n+1) \cdot   \deg \cdot  \ell \cdot \basis^2 \cdot \bound_{\ma{C}} + \max(\bound_{\vec{g}}, \bound_{\vec{h}})$.
\end{lemma}

\textbf{Modulus Switching.} 
The modulus switching technique, developed by
Brakerski and Vaikuntanathan \cite{FOCS:BraVai11}, allows an evaluator to change the
modulus of a given ciphertext without the knowledge of the secret key. 

\begin{definition}[Modulus Switching for LWE Samples] 
We define modulus switching  from $\ZZ_Q$ to $\ZZ_q$ by the following algorithm.  
\begin{itemize}
\item $\ModSwitch(\vec{c}, Q, q)$:  
On input a LWE sample $\vec{c} = [b, \vec{a}^{\top}]^{\top} \in \LWE_{\bound, n, Q}(\vec{s}, .)$, and moduli $Q$ and $q$,
the function outputs $\ct_{\out} \in \LWE_{\bound, n, q}(\vec{s}, .)$.
\end{itemize}
\end{definition}

\begin{lemma}[Correctness of Modulus Switching]
\label{lemma:correctness-modulus-switching} 
Let $\vec{c} = \LWE_{\bound_{\vec{c}}, n, Q}(\vec{s}, \msg \cdot (\frac{Q}{t}))$
where
$\vec{s} \in \ZZ_q^{n}$ and
$Q = 0 \mod t$. 
If $\vec{c}_{\out} \exec \ModSwitch(\vec{c}, Q, q)$,
then $\vec{c}_{\out} \in \LWE_{\bound, n, q}(\vec{s}, \msg \cdot (\frac{q}{t}))$,  
where 
\begin{align*}
\bound \leq \frac{q^2}{Q^2} \cdot \bound_{\vec{c}} + \frac{2}{3} + \frac{2}{3} \cdot \hamming(\vec{s}) \cdot \big(\norm{\Var(\vec{s})}_{\infty} + \norm{\E(\vec{s})}_{\infty}^2\big).
\end{align*} 
\end{lemma}

\textbf{Sample Extraction.} 
Informally, sample extraction allows extracting from an RLWE sample an LWE sample of a single coefficient without increasing the error rate. Here we give a generalized version of the extraction algorithm, which can extract an LWE sample of any coefficient of the RLWE sample.

\begin{definition}[Sample Extraction]
Sample extraction consists of algorithms $\KeyExtract$ and $\SampleExtract$, defined as follows.
\begin{itemize}
\item $\KeyExtract(\fr{s})$:
Takes as input a key $\fr{s} \in \R_{\deg, q}$, and outputs $\coefs(\fr{s})$.

\item $\SampleExtract(\ct, k)$: 
Takes as input $\ct = [\fr{b}, \fr{a}]^{\top} \in \RLWE_{\bound, \deg, q}(\fr{s}, .)$  and an index $k \in [\deg]$.
The function outputs $\vec{c} \in \LWE_{\bound, \deg, q}(\vec{s}, .)$.
\end{itemize} 
\end{definition}
 
\begin{lemma}[Correctness of Sample Extraction]
\label{lemma:correctness-sample-extraction}
Let  $\ct \in  \RLWE_{\bound_{\ct}, \deg, q}(\fr{s}, \msg)$, where $\fr{s}, \msg \in \R_{\deg, q}$. 
Denote $\msg = \sum_{i=1}^{\deg} \msg_i \cdot X^{i - 1}$, where $\msg_{i} \in \ZZ_q$.
If $\vec{s} \exec \KeyExtract(\fr{s})$ and $\vec{c} \exec \SampleExtract(\ct, k)$,
for some $k \in [\deg]$, 
then $\vec{c} \in \LWE_{\bound, \deg, q}(\vec{s}, \msg_{k})$,
where  
$\bound \leq \bound_{\ct}$.
\end{lemma}

\textbf{Key Switching.} 
 Key-switching is an important technique to build scalable homomorphic encryption schemes.  In short, having a key switching key, the evaluator can map a given LWE sample to an LWE sample of a different key.

\begin{definition}[Key Switching LWE to GLWE]\label{def:packing-lwe-to-rlwe-ordinary}
We define the key generation algorithm $\KeySwitchSetup$ and  LWE to GLWE switching $\KeySwitch$ as follows.

\begin{itemize} 
\item $\KeySwitchSetup(\bound_{\keySwitchKey}, n, n', \deg, q, \vec{s}, \vec{s}', \basis_{\keySwitchKey})$:
Takes as input 
bound $\bound_{\keySwitchKey} \in \NN$, dimensions $n, n' \in \NN$, a degree $\deg \in \NN$ and
a modulus $q \in \NN$,  vectors  $\vec{s} \in \ZZ_q^n$, $\vec{s}' \in \R_{\deg, q}^{n'}$ and a basis $\basis_{\keySwitchKey} \in \NN$.  
The algorithm outputs $\keySwitchKey \in \GLWE_{\bound_{\keySwitchKey}, n', \deg, q}(\vec{s}', .)^{n \times \ell_{\keySwitchKey}}$.
%

\item $\KeySwitch(\vec{c}, \keySwitchKey)$: 
Takes as input  $\vec{c} = [b, \vec{a}^{\top}]^{\top} \in \LWE_{\bound_{\vec{c}}, n, q}(\vec{s}, .)$
and a key switching key $\keySwitchKey \in \GLWE_{\bound_{\keySwitchKey}, n', \deg, q}(\vec{s}', .)^{n \times \ell_{\keySwitchKey}}$.
The algorithm outputs $c_{\out} \in \GLWE_{\bound, \deg, q}(\vec{s}', .)$.
\end{itemize} 
 \end{definition}

\begin{lemma}[Correctness of Key Switching]
\label{lemma:correctness-lwe-to-rlwe-packing}
If $\keySwitchKey \exec \KeySwitchSetup(\bound_{\keySwitchKey}, n, n', \deg, q, \vec{s}, \vec{s}', \basis_{\keySwitchKey})$  
and $c_{\out} \exec \KeySwitch(\vec{c}, \keySwitchKey)$  
then $c_{\out} \in \GLWE_{\bound, n', \deg, q}(\vec{s}', \msg)$,
where  
\begin{align*}
\bound \leq   \bound_{\vec{c}} + (\frac{1}{3}) \cdot n \cdot \ell_{\keySwitchKey} \cdot \basis_{\keySwitchKey}^2 \cdot  \bound_{\keySwitchKey}
\end{align*} 
%
and
$\ell_{\keySwitchKey} = \roundUp{\log_{\basis_{\keySwitchKey}}{q}}$.
\end{lemma}

\textbf{TFHE Blind Rotation and Bootstrapping.} 
 Below we recall the blind rotation introduced by Chillotti et al. \cite{AC:CGGI16}.
 For both blind rotation and bootstrapping we use a decomposition vector
   $\vec{u} \in \ZZ^u$ for which the following holds.
  For all $y \in \S$ with $\S \subseteq \ZZ_{2 \cdot \deg}$ there exists $\vec{x} \in \BB^u$ such that $y = \sum_{i=1}^u \vec{x}[i] \cdot \vec{u}[i] \mod 2 \cdot \deg$. For example for  $\S = \{0, 1\}$ we have $\vec{u} = \{1\}$, and for $\S = {-1, 0, 1}$ we have $\vec{u} = \{-1, 1\}$.

\begin{definition}[TFHE Blind Rotation]
Blind rotation consists of a key generation algorithm $\BRKeyGen$,
and blind rotation algorithm $\BlindRotate$.

\begin{itemize}

\item $\BRKeyGen(n, \vec{s}, \vec{u}, \bound_{\brKey}, \deg, Q, \basis_{\RGSW}, \fr{s})$: 
Takes as input a dimension $n$, a secret key $\vec{s} \in \ZZ_t^{n \times 1}$ and a vector $\vec{u} \in \ZZ^{u}$, a bound $\bound_{\brKey}$, a degree $\deg$ and a modulus $Q$ defining the ring $\R_{\deg, Q}$, 
 a basis $\basis_{\RGSW} \in \NN$, and a secret key $\fr{s} \in \R_{\deg, Q}$.
The algorithm outputs a blind rotation key $\brKey$.

\item $\BlindRotate(\brKey, \acc, \ct, \vec{u})$: 
Takes as input a blind rotation key $\brKey$ $=$ $\RGSW_{\bound_{\brKey}, \deg, Q, \basis_{\RGSW}}(\fr{s}, .)^{n \times u}$,  
  an accumulator $\acc \in \RLWE_{\bound_{\acc}, \deg, Q}(\fr{s}, .)$, 
 a ciphertext $\ct \in  \LWE_{\bound_{\ct}, n, 2\cdot \deg}(\vec{s}, .)$, where $\vec{s} \in \ZZ_t^{n}$ with $t \leq 2 \cdot \deg$, and a vector $\vec{u} \in \ZZ^{u}$. The algorithm outputs $\acc_{\out}$.
\end{itemize} 
\end{definition}

\begin{lemma}[Correctness of TFHE-Style Blind-Rotation]
\label{lemma:correctness-tfhe-blind-rotate} 
 Let   $\ct$ $\in$ $\LWE_{\bound_{\ct}, n, 2 \cdot \deg}(\vec{s}$, $.)$ be a LWE sample and denote $\ct$ $=$ $[b$, $\vec{a}^{\top}]^{\top} \in \ZZ_{2 \cdot \deg}^{(n+1) \times 1}$.
Let $\acc$ $\in$ $\RLWE_{\bound_{\acc}, \deg, Q}(\fr{s}$, $\msg_{\acc})$. 
If $\brKey \exec \BRKeyGen(n, \vec{s}, \vec{u}, \bound_{\brKey}, \deg, Q, \basis_{\RGSW}, \fr{s})$
and $\acc_{\out}$ $\exec$ $\BlindRotate(\brKey$, $\acc$, $\ct$, $\vec{u})$,
then $\acc_{\out}$ $\in$ $\RLWE_{\bound_{\out}, \deg, Q}(\fr{s}$, $\msg_{\acc}^{\out})$,
where $\msg_{\acc}^{\out}$ $=$ $\msg_{\acc}$ $\cdot$ $X^{b - \vec{a}^{\top} \cdot \vec{s} \mod 2 \cdot \deg}$ $\in$ $\R_{\deg, Q}$ and  
\begin{align*} 
\bound_{\out}  \leq  \bound_{\acc} +  (\frac{2}{3}) \cdot  n \cdot u \cdot \deg \cdot \ell_{\RGSW} \cdot \basis_{\RGSW}^2 \cdot \bound_{\brKey}
\end{align*} 
\end{lemma}

 \begin{definition}[Bootstrapping]\label{def:tfhe-bootstrapping}
The bootstrapping procedure is as follows:
\begin{itemize}
\item $\Bootstrap(\brKey,  \vec{u}, \ct, \preRotPoly, \keySwitchKey)$: 
Takes as input a blind rotation key $\brKey$ $=$  $\RGSW_{\bound_{\brKey}, \deg, Q, \basis_{\brKey}}(\fr{s}$, $.)^{n \times u}$, 
a vector $\vec{u} \in \ZZ^{u}$,
a LWE sample $\ct$ $=$ $\LWE_{\bound_{\ct}, n, q}(\vec{s}$, $.)$ $=$ $[b$, $\vec{a}^{\top}]^{\top} \in \ZZ_q^{(n+1) \times 1}$, 
a polynomial  $\preRotPoly \in \R_{\deg, Q}$, and a LWE to LWE key switching key $\keySwitchKey$.
The algorithm outputs $\ct_{\out}$.
%
%
%
%
%
%
%
%
\end{itemize}
\end{definition}

  \begin{theorem}[Correctness of TFHE-Style Bootstrapping]
\label{thm:correctness-tfhe-bootstrapping}
 Let   $\ct \in \LWE_{\bound_{\ct}, n, q}(\vec{s}, \Delta_{q, t} \cdot \msg)$.
Let $\brKey$ $\exec$ $\BRKeyGen(n$, $\vec{s}$, $\vec{u}$, $\bound_{\brKey}$, $\deg$, $Q$, $\basis_{\RGSW}$, $\fr{s})$,
$\vec{s}_{F} \exec \KeyExtract(\fr{s})$, 
and
$\keySwitchKey$ $\exec$ $\KeySwitchSetup(\bound_{\keySwitchKey}$, $\deg$, $n$, $q$, $\vec{s}_{F}$, $\vec{s}$, $\basis_{\keySwitchKey})$.
Let $\ct_{\deg} \exec \ModSwitch(\ct, q, 2 \cdot \deg) \in \LWE_{\bound_{\deg}, n, 2 \cdot \deg}(\vec{s}, \Delta_{2 \cdot \deg, t} \cdot \msg)$.
Let $\tilde{\msg} \exec \Phase(\ct_{\deg})$, $\msg = \round{\tilde{\msg}}_{t}^{2 \cdot \deg}$ and
 $\preRotPoly \in \R_{\deg, Q}$ be such that if $\tilde{\msg} \in [0,\deg)$, then
 $\coefs(X^{\tilde{\msg}} \cdot \preRotPoly)[1] = \Delta_{Q, t} \cdot F(\msg)$, where $F:\ZZ_t \mapsto \ZZ_t$.
If $\ct_{\out} \exec \Bootstrap(\brKey,  \vec{u}, \ct, \preRotPoly, \keySwitchKey)$
and $\msg \leq \frac{t}{2} - 1$,
then $\ct_{\out} \in \RLWE_{\bound_{\out}, \deg, q}(\vec{s}, \msg_{\out})$,
where $\msg_{\out} = \Delta_{Q, t} \cdot F(\msg)$ and  
\begin{align*}
\bound_{\out} \leq& (\frac{q^2}{Q^2}) \cdot  (\bound_{\BR} + \bound_{KS}) \\
 &+ \frac{2}{3} + \frac{2}{3} \cdot \hamming(\vec{s}_{F}) \cdot  (\norm{\Var(\vec{s}_{F})}_{\infty} + \norm{\E(\vec{s}_{F})}_{\infty}^2), 
\end{align*}
where  $\bound_{\BR} \leq (\frac{2}{3}) \cdot  n \cdot u \cdot \deg \cdot \ell_{\RGSW} \cdot \basis_{\RGSW}^2 \cdot \bound_{\brKey}$
and $\bound_{KS} \leq  (\frac{1}{3}) \cdot \deg \cdot \ell_{\keySwitchKey} \cdot \basis_{\keySwitchKey}^2 \cdot  \bound_{\keySwitchKey}$.
If $Q = q$, then $\bound_{\out} \leq \bound_{\BR} + \bound_{KS}$.
Finally we have,
\begin{align*}
\bound_{\deg} \leq& (\frac{2 \cdot \deg}{q})^2 \cdot \bound_{\ct} \\
 &+ \frac{2}{3}  + \frac{2}{3} + \hamming(\vec{s}) \cdot (\norm{\Var(\vec{s})}_{\infty} + \norm{\E(\vec{s})}_{\infty}^2)
\end{align*} 
if $\deg \neq 2 \cdot q$, and $\bound_{\deg} = \bound_{\ct}$ otherwise, where
$\bound_{\ct} = \max(\bound_{\fresh}, \bound_{\out})$ and $\bound_{\fresh}$ is the bound in the error variance
of a fresh ciphertext.
\end{theorem}


\section{Our Functional Bootstrapping Technique}\label{sec:functional-bootstrapping}

\begin{figure*}[t]
\center
\fbox{
\procedure[syntaxhighlight=auto ]{$\PubMux({[\vec{c}_i]}_{i=1}^{\ell}, \fr{p}_0, \fr{p}_1)$}{  
    \textbf{Input:~} \text{Takes as input a vector $[\vec{c}_i]_{i=1}^{\ell}$ } \\
   \text{where $\vec{c}_i$ $\in$ $\GLWE_{\bound, n, \deg, q}(\vec{s}$, $\msg \cdot \basis^{i-1})$} \\
\text{where $\msg \in \{0, \Delta_{q, t}\}$ and $\ell = \roundUp{\log_{\basis} q}$, and $\fr{p}_0$, $\fr{p}_1 \in \R_{\deg, q}$.} \\
\text{Note that when $\deg = 1$, i.e., $\vec{c}_i$ are LWE samples,} \\
 \text{we have $\fr{p}_0$, $\fr{p}_1 \in \ZZ_q$.}
 \\[][\hline] 
\pcln \text{Set $\fr{p} = \fr{p}_1 - \fr{p}_0$.}  \\
\pcln \text{Set $\vec{p} = \Decomp_{\basis, q}(\fr{p})$.} \\ 
\pcln \text{Compute 
$\vec{c}_{\out} \exec  
[\Delta_{q, t} \cdot \fr{p}_0,~~\vec{0}]^{\top} + \sum_{i=1}^{\ell} \vec{p}[i] \cdot \vec{c}_i$.} \\
\pcln \text{Return $\vec{c}_{\out}$.}
}  
  
\pchspace[0.2cm]

\procedure[syntaxhighlight=auto ]{$\BuildAcc({[\acc_i]}_{i=1}^{\ell}, \fr{p}_0, \fr{p}_1, \keySwitchKey)$}{  
   \textbf{Input:~} 
  \text{Takes as input a vector $[\acc_i]_{i=1}^{\ell}$} \\
   \text{where $\acc_i \in \LWE_{\bound_{\acc}, n, q}(\vec{s}, \Delta_{q, t} \cdot \msg \cdot \basis^{i-1})$,} \\
\text{where $\msg \in \{0, 1\}$ and $\ell = \roundUp{\log_{\basis} q}$, $\fr{p}_0, \fr{p}_1 \in \R_{\deg, q}$} \\
\text{and a key switching key $\keySwitchKey$.}
 \\[][\hline] 
\pcln  \text{For all $i \in [\ell]$ compute $\acc_{R, i} \exec \KeySwitch(\acc_i, \keySwitchKey)$.} \\
\pcln \text{Compute $\vec{c}_{\out} \exec \PubMux([\acc_{R, i}]_{i=1}^{\ell}, \fr{p}_0, \fr{p}_1)$.} \\
\pcln  \text{Output $\vec{c}_{\out}$.}
}  
} 
\caption{On the left we show the public Mux algorithm. On the right is the accumulator builder algorithm.}
\label{figure:cmux-accumulator-builder}
\end{figure*}

\begin{figure*}[t]
\center
\fbox{
\procedure[syntaxhighlight=auto ]{$\Bootstrap(\brKey,  \vec{u}, \ct, \preRotPoly_{0}, \preRotPoly_{1},  \basis_{\boot}, \packKey, \keySwitchKey)$}{  
    \textbf{Input:~} 
    \text{A bootstrapping key $\brKey = \RGSW_{\bound_{\brKey}, \deg, Q, \basis_{\brKey}}(\fr{s}, .)^{n \times u}$,} \\
\text{vector $\vec{u} \in \ZZ^{u}$,} \\
 \text{LWE sample $\ct = \LWE_{\bound_{\ct}, n, q}(\vec{s}, .)(\msg) = [b, \vec{a}^{\top}]^{\top} \in \ZZ_q^{(n+1) \times 1}$,} \\ \text{polynomials  $\preRotPoly_{0}, \preRotPoly_{1}, \in \R_{\deg, Q}$,} \\
 \text{decomposition basis $\basis_{\boot}$,} \\
 \text{LWE to RLWE key switching key $\packKey$, } \\
 \text{and a LWE to LWE key switching key $\keySwitchKey$.}
 \\[][\hline] 
\pcln \text{Set $\sgnRotPoly = 1 - \sum_{i=2}^{\deg} X^{i-1} \in \R_{\deg, Q}$.} \\
\pcln \text{Set $\ct_{\deg} \exec \ModSwitch(\ct, q, 2 \cdot \deg) = [b_{\deg}, \vec{a}_{\deg}^{\top}]^{\top}$.} \\
\pcln \text{Compute $\sgnRotPoly_b \exec \sgnRotPoly \cdot X^{b_\deg} \in \R_{\deg, Q}$.} \\
\pcln \text{Let $\ell_{\boot} = \roundUp{\log_{\basis_{\boot}}{Q}}$.} \\
\pcln \text{For $i=1$ to $\ell_{\boot}$ do,} \\ 
\pcln\pcind \text{Set $\acc_i \exec [\basis_{\boot}^{i-1} \cdot \sgnRotPoly_b \cdot (\frac{\Delta_{Q, t}}{2}), 0]^{\top}$.} \\ 
\pcln\pcind \text{Run $\acc_{\BR, i} \exec \BlindRotate(\brKey, \acc_i, \ct_{\deg}, \vec{u})$. } \\
\pcln\pcind \text{Run $\acc_{\vec{c}, i} \exec \SampleExtract(\acc_{\BR, i}, 1) + [\basis_{\boot}^{i-1} \cdot \frac{\Delta_{Q, t}}{2}, \vec{0}]$.}\\
\pcln \text{Set $\preRotPoly_{0}' \exec \preRotPoly_{0} \cdot X^{b_{\deg}} \in \R_{\deg, Q}$.} \\
\pcln \text{Set $\preRotPoly_{1}' \exec \preRotPoly_{1} \cdot X^{b_{\deg}} \in \R_{\deg, Q}$.} \\
\pcln \text{Run $\acc_{F} \exec \BuildAcc([\acc_{\vec{c}, i}]_{i=1}^{\ell_{\boot}}, \preRotPoly'_0, \preRotPoly'_1, \packKey)$.} \\
\pcln \text{Run $\acc_{\BR, F} \exec \BlindRotate(\brKey, \acc_{F}, \ct_{\deg}, \vec{u})$.} \\   
\pcln \text{Run $\vec{c}_{Q} \exec \SampleExtract(\acc_{\BR, F}, 1)$.} \\
\pcln \text{Run $\vec{c}_{Q, \keySwitchKey} \exec \KeySwitch(\vec{c}_{Q}, \keySwitchKey)$.} \\
\pcln \text{Return $\ct_{\out} \exec \ModSwitch(\vec{c}_{Q, \keySwitchKey}, Q, q)$.}
}

\pchspace[0.2cm] 

\procedure[syntaxhighlight=auto ]{$\ConstPoly(\BootsFunction, \deg)$}{  
   \textbf{Input:~}  
\text{A lookup table $\BootsFunction:\ZZ_t \rightarrow \ZZ_Q$} \\
 \text{that maps elements in $\ZZ_t$ to elements in $\ZZ_Q$,} \\
 \text{and the degree that defines the ring $\R_{\deg, Q}$.}
 \\[][\hline]  
\pcln \text{For $\msg \in [0,t-1]$ and for $e \in \bigg[-\big\lceil\frac{\deg}{t}\big\rceil, \big\lceil\frac{\deg}{t}\big\rceil\bigg]$ do}\\
\pcln\pcind \text{Set $y \exec \big\lfloor\frac{2 \cdot N}{t}\big\rceil \cdot \msg + e \pmod{2N}$.} \\
\pcln\pcind \text{If $y=0$,} \\
\pcln\pcind\pcind \text{set $\vec{p}_0[1] \exec \BootsFunction[\msg+1]$.} \\
\pcln\pcind \text{If $y \in [1,\deg-1]$,} \\
\pcln\pcind\pcind \text{set $\vec{p}_0[\deg - y + 1] \exec -\BootsFunction[\msg+1] \in \ZZ_Q$.} \\
\pcln\pcind \text{Set $y' \exec y - \deg$.} \\
\pcln\pcind \text{If $y = \deg$,} \\
\pcln\pcind\pcind \text{set $\vec{p}_1[1] \exec -\BootsFunction[\msg+1]$.} \\
\pcln\pcind \text{If $y \in [\deg+1,2\cdot \deg-1]$,} \\
\pcln\pcind\pcind \text{set $\vec{p}_1[\deg - y' + 1] \exec \BootsFunction[\msg+1] \in \ZZ_Q$.} \\
\pcln \text{Set $\preRotPoly_0 = \sum_{i=1}^{\deg} \vec{p}_0[i] \cdot X^{i-1}$.}\\
\pcln \text{Set $\preRotPoly_1 = \sum_{i=1}^{\deg} \vec{p}_1[i] \cdot X^{i-1}$.} \\
\pcln \text{Return $(\preRotPoly_0, \preRotPoly_1)$.}
}  
} 
\caption{On the left is our full domain bootstrapping algorithm. 
On the right is the algorithm for setting up the rotation polynomials based on a lookup table.}
\label{figure:functional-bootstrapping}
\end{figure*}

In this section, we show our functional bootstrapping algorithm.
First, however, we introduce two sub-procedures that help us construct the bootstrapping procedure. 

\subsection{Public Mux Gate and Building a Homomorphic Accumulator}

We introduce a public version of the CMux gate from \cite{AC:CGGI16} which aims to choose one of two given polynomial plaintexts based on an encrypted bit. 
Furthermore, we show how to use our public Mux gate to build an accumulator for the functional bootstrap. 
Roughly speaking, the accumulator builder gets
an encrypted bit and based on that bit chooses one of two different rotation polynomials.
The difference between the accumulator builder and the public mux is that our accumulator builder needs to switch the encrypted bit from LWE samples to RLWE samples.
The algorithms are depicted in Figure~\ref{figure:cmux-accumulator-builder}. 
Below we state and prove the correctness theorems for both algorithms.

%
%
%
%

\begin{lemma}[Correctness of Public Mux] 
Let $[\vec{c}_i]_{i=1}^{\ell}$ be a vector where $\vec{c}_i$ $\in$ $\GLWE_{\bound_{\vec{c}}, n, \deg, q}(\vec{s}$, $\msg \cdot \Delta_{q, t} \cdot \basis^{i-1})$
where $n, \deg, q \in \NN$, $\msg \in \{0, 1\}$,  $\ell = \roundUp{\log_{\basis} q}$. 
Let $\fr{p}_0$, $\fr{p}_1 \in \R_{\deg, q}$.
If $\vec{c}_{\out} \exec \PubMux({[\vec{c}_i]}_{i=1}^{\ell}, \fr{p}_0, \fr{p}_1)$, then 
$\vec{c}_{\out} \exec \GLWE_{\bound_{\out}, n, \deg, q}(\vec{s}, \Delta_{q, t} \cdot \fr{p}_{\msg})$,
where $\bound_{\out} \leq (\frac{1}{3}) \cdot \deg \cdot \ell \cdot \basis^2 \cdot \bound_{\vec{c}}$.
If $\fr{p}_0$, $\fr{p}_1$ are monomials then $\deg$ disappears from the above inequality.
\end{lemma}

\begin{proof}
From the correctness of linear homomorphism of GLWE samples, we have
\begin{align*}
\vec{c}_{\out} &=  
\begin{bmatrix}
\Delta_{q, t} \fr{p}_0 \\
\vec{0}
\end{bmatrix} + \sum_{i=1}^{\ell} \vec{p}[i] \cdot \GLWE_{\bound_{\vec{c}}, n, \deg, q}(\vec{s}, \msg  \Delta_{q, t}  \basis^{i-1}) \\
&= \begin{bmatrix}
\Delta_{q, t}  \fr{p}_0 \\
\vec{0}
\end{bmatrix} + \GLWE_{\bound_{\out}, n, \deg, q}(\vec{s},  \Delta_{q, t} \msg   (\fr{p}_1 - \fr{p}_0)) \\
&= \GLWE_{\bound_{\out}, n, \deg, q}(\vec{s}, \Delta_{q, t} \cdot \fr{p}_{\msg}),
\end{align*}
where $\bound_{\out} \leq (\frac{1}{3}) \cdot \deg \cdot \ell \cdot \basis^2 \cdot \bound_{\vec{c}}$.
If $\fr{p}_0$, $\fr{p}_1$ are monomials then $\deg$ disappears from the above inequality.
\end{proof}

%
%
%
%

 \begin{lemma}[Correctness of the Accumulator Builder]\label{lemma:correctness-accumulator-builders}
Let
 $[\acc_i]_{i=1}^{\ell}$ where $\acc_i \in \LWE_{\bound_{\acc}, n, q}(\vec{s}, \Delta_{q, t} \cdot \msg \cdot \basis^{i-1})$,
where $\msg \in \{0, 1\}$, $n, q  \in \NN$, $\bound_{\acc} \le q$, $\vec{s} \in \ZZ_q^{n \times 1}$, and
$\ell = \roundUp{\log_{\basis} q}$. 
Let $\keySwitchKey \exec \KeySwitchSetup(\bound_{\keySwitchKey}, n, 1, \deg, q, \vec{s}, \fr{s}, \basis_{\keySwitchKey})$,
where $\deg \in \NN$, $\fr{s} \in \R_{\deg, q}$ and $\basis_{\keySwitchKey} \le q$.
Finally let  $\fr{p}_0, \fr{p}_1 \in \R_{\deg, q}$.
If $\vec{c}_{\out} \exec \BuildAcc({[\acc_i]}_{i=1}^{\ell}, \fr{p}_0, \fr{p}_1, \keySwitchKey)$,
then $\vec{c}_{\out} \in \RLWE_{\bound_{\out}, \deg, q}(\fr{s}, \Delta_{q, t} \cdot \fr{p}_{\msg})$,
where 
\begin{align*}
\bound_{\out} \leq (\frac{1}{3}) \cdot \deg \cdot \ell \cdot \basis^2 \cdot \big(\bound_{\acc} + (\frac{1}{3}) \cdot n \cdot \ell_{\keySwitchKey} \cdot \basis_{\keySwitchKey}^2 \cdot \bound_{\keySwitchKey}\big).
\end{align*} 
If $\fr{p}_0, \fr{p}_1$ are monomials, then the degree $\deg$ disappears from the above inequality.
 \end{lemma}

 \begin{proof}
 From correctness of $\KeySwitch$ we have that 
 $\acc_{R, i} \in \RLWE_{\bound_{R}, \deg, q}(\fr{s}, \Delta_{q, t} \cdot \msg \cdot \basis^{i-1})$,
 where $\bound_{R} \leq \bound_{\acc} + \frac{1}{3} \cdot n \cdot \ell_{\keySwitchKey} \cdot \basis_{\keySwitchKey}^2 \cdot \bound_{\keySwitchKey}$.  
 From correctness of $\PubMux$ we have that $\vec{c}_{\out} \in \RLWE_{\bound_{\out}, \deg, q}(\fr{s}, \Delta_{q, t} \cdot \fr{p}_{\msg})$,
 where $\bound_{\out} \leq (\frac{1}{3}) \cdot \deg \cdot \ell \cdot \basis^2 \cdot \bound_{R}$.
 Also from correctness of $\PubMux$ we have that if $\fr{p}_0, \fr{p}_1$ are monomials, then the degree $\deg$ disappears from the above inequalities. 
\end{proof}

\subsection{The Functional Bootstrapping Algorithm} 
 
Our bootstrapping algorithm depicted in Figure~\ref{figure:functional-bootstrapping}.
There are three phases. First, we compute an extended ciphertext whose plaintext indicates whether the given ciphertexts phase is in $[0, \deg)$ or in $[\deg, 2 \cdot \deg)$.
Based on this ciphertext, we build the accumulator.
The final step is to blind rotate the rotation polynomial and switch
the ciphertext from RLWE to LWE. 
Remind that that the vector $\vec{u} \in \ZZ^u$ is such
that for all $y \in \S$ with $\S \subseteq \ZZ_{2 \cdot \deg}$ there exists $\vec{x} \in \BB^u$ such that $y = \sum_{i=1}^u \vec{x}[i] \cdot \vec{u}[i] \mod 2 \cdot \deg$.
Furthermore, we denote as $\basis_{\RGSW}$, $\basis_{\keySwitchKey}$ and $\basis_{\packKey}$
the decomposition basis for $\RGSW$ samples, and two different key switch keys.
  Below we state the correctness theorem.

\begin{theorem}[Correctness of the Bootstrapping]\label{thm:correctness-of-functional-bootstrapping}  
Let $\vec{s} \in \ZZ_q^{n}$ for some $q, n \in \NN$, $\fr{s} \in \R_{\deg, q}$,   
and $\vec{s}_{F} \exec \KeyExtract(\fr{s})$. 
Furthermore, let us define the following. 
\begin{itemize}
  
\item $\brKey \exec \BRKeyGen(n, \vec{s}, \vec{u}, \bound_{\brKey}, \deg, Q, \basis_{\RGSW}, \fr{s})$.

\item $\packKey \exec \KeySwitchSetup(\bound_{\packKey},\deg,1, \deg, Q, \vec{s}_{F}, \fr{s}, \basis_{\packKey})$.

\item $\keySwitchKey \exec \KeySwitchSetup(\bound_{\keySwitchKey}, \deg, n,1, Q, \vec{s}_{F}, \vec{s}, \basis_{\keySwitchKey})$.

\item  $\ct \in \LWE_{\bound_{\ct}, n, q}(\vec{s}, \Delta_{q, t} \cdot \msg)$.

\item $\ct_{\deg} \exec \ModSwitch(\ct, q, 2 \cdot \deg)$, and let us denote $\ct_{\deg} \in \LWE_{\bound_{\deg}, n, 2 \cdot \deg}(\vec{s}, \Delta_{2 \cdot \deg, t} \cdot \msg)$.


\end{itemize}

Finally, let $\tilde{\msg} \exec \Phase(\ct_{\deg})$, $\round{\tilde{\msg}}_{t}^{2 \cdot \deg} = \msg$ and
 $\preRotPoly_0, \preRotPoly_1 \in \R_{\deg, Q}$ be such that 
\begin{itemize}
\item
if $\tilde{\msg} \in [0,\deg)$, then $\coefs(X^{\tilde{\msg}} \cdot \preRotPoly_0)[1] = \Delta_{Q, t'} \cdot f(\msg)$, and

\item  
if $\tilde{\msg} \in [\deg, 2 \cdot \deg)$, then
$\coefs(X^{\tilde{\msg}} \cdot \preRotPoly_1)[1] = \Delta_{Q, t'} \cdot f(\msg)$, where $f:\ZZ_t \mapsto \ZZ_t'$.
\end{itemize}

If $\ct_{\out} \exec \Bootstrap(\brKey, \vec{u}, \ct, \preRotPoly_0, \preRotPoly_1,  \basis_{\boot}, \packKey, \keySwitchKey)$ and given that
$\bound_{\deg} \le \frac{\deg}{t}$, 
then $\ct_{\out} \in \LWE_{\bound_{\out}, n, q}(\vec{s}, \Delta_{q, t'} \cdot f(\msg))$,
where   
 \begin{align*}
 \bound_{\deg} \leq& (\frac{2 \cdot \deg}{q})^2 \cdot \bound_{\ct} \\
  &+ \frac{2}{3}  + \frac{2}{3} + \hamming(\vec{s}) \cdot (\norm{\Var(\vec{s})}_{\infty} + \norm{\E(\vec{s})}_{\infty}^2),
 \end{align*}
for $\deg \neq 2 \cdot q$, and $\bound_{\deg} = \bound_{\ct}$ otherwise, and $\bound_{\ct} \leq \max(\bound_{\fresh}, \bound_{\out})$.
We use the bound $\bound_{\fresh}$ if $\ct$ is a freshly encrypted ciphertext.
Finally the bound $\bound_{\out}$
of the bootstrapped ciphertext $\ct_{\out}$ that is as follows
\begin{align*}
\bound_{\out} \leq& \frac{q^2}{Q^2} \cdot (\bound_{F} + \bound_{\BR} + \bound_{KS}) \\
 &+ \frac{2}{3} + \frac{2}{3} \cdot \hamming(\vec{s}_F) \cdot (\norm{\Var(\vec{s}_F)}_{\infty} + \norm{\E(\vec{s}_F)}_{\infty}^2), 
\end{align*}
for $Q \neq q$, and $\bound_{\out} \leq \bound_{F} + \bound_{\BR} + \bound_{KS}$ for $Q=q$,
where 
\begin{itemize}
\item $\bound_{F} \leq (\frac{1}{3}) \cdot \deg  \cdot \ell_{\boot} \cdot  \basis_{\boot}^2  \cdot \big(\bound_{\BR} +  \bound_{P}\big)$

\item $\bound_{\BR} \leq (\frac{2}{3}) \cdot n \cdot u \cdot  \deg \cdot \ell_{\RGSW} \cdot \basis_{\RGSW}^2 \cdot \bound_{\brKey}$,

\item $\bound_{KS} \leq (\frac{1}{3}) \cdot \deg \cdot \ell_{\keySwitchKey} \cdot \basis_{\keySwitchKey}^2 \cdot \bound_{\keySwitchKey}$, and
 
\item $\bound_{P} \leq (\frac{1}{3}) \cdot \deg \cdot \ell_{\packKey} \cdot \basis_{\packKey}^2 \cdot \bound_{\packKey}$. 
 
\end{itemize}

\end{theorem}

\begin{proof} 
Let us denote  $\Phase(\ModSwitch(\ct, q, 2 \deg)) = \tilde{\msg} = \Delta_{2 \cdot \deg, t} \cdot \msg + e \mod 2 \deg$. 
From correctness of modulus switching we have $\norm{\Var(e)}_{\infty} \leq \bound_{\deg} \leq (\frac{2 \cdot \deg}{q})^2 \cdot \bound_{\ct} + \frac{2}{3}  + \frac{2}{3} + \hamming(\vec{s}) \cdot (\norm{\Var(\vec{s})}_{\infty} + \norm{\E(\vec{s})}_{\infty}^2)$. 
The first $\ell_{\boot}$ blind rotations compute powers of $\basis_{\boot}$ of the sign function.
Specifically from the correctness of blind rotation we have 
$\acc_{\BR, i} \in \RLWE_{\bound_{\BR}, \deg, Q}(\fr{s}, \basis_{\boot}^{i-1} \cdot \frac{\Delta_{Q, t}}{2} \cdot \sgnRotPoly \cdot X^{\tilde{\msg}})$.
Hence, if $\tilde{\msg} \mod 2 \deg \in [0, \deg)$, then the constant coefficient
is equal to $\basis_{\boot}^{i-1} \cdot \frac{\Delta_{Q, t}}{2} \in \ZZ_Q$.
If $\tilde{\msg} \mod 2  \deg \in [\deg, 2 \cdot \deg)$ then
the constant coefficient is equal to $- \basis_{\boot}^{i-1} \cdot \frac{\Delta_{Q, t}}{2} \in \ZZ_Q$.
Equivalently, the constant coefficient is $\sgn(\msg) \cdot \basis_{\boot}^{i-1} \cdot \frac{\Delta_{Q, t}}{2} \in \ZZ_Q$.
Note that $\sgn(\tilde{\msg}) \mod 2\cdot N = \sgn(\msg) \mod t$.
From correctness of sample 
extraction and additive homomorphism of LWE samples we have 
$\acc_{\vec{c}, i} \in \LWE_{\bound_{\BR}, \deg, Q}(\vec{s}_{\sgn}, \msg_{\BR, i})$,
where   
$\bound_{\BR}$ is the error stemming from the blind rotation, and $\msg_{\BR, i} = \sgn(\msg) \cdot \basis_{\boot}^{i-1} \cdot \frac{\Delta_{Q, t}}{2} + \basis_{\boot}^{i-1} \cdot \frac{\Delta_{Q, t}}{2}$.
Note that if $\sgn(\msg) = 1$, then $\msg_{\BR, i} = \basis_{\boot}^{i-1} \cdot \Delta_{Q, t}$,
and if $\sgn(\msg) = -1$, then $\msg_{\BR, i} = 0$.

 The next step is to create a new accumulator holding the polynomial 
 $\preRotPoly_{\inter} \in \R_{\deg, Q}$, where $\inter = 0$ if $\msg + e  \in [0,\deg)$ and
 $\inter=1$ if $\msg + e  \in [\deg,2\deg)$.
 From correctness of the accumulator builder we have the 
 $\acc_{F} \in \RLWE_{\bound_{F}, \deg, Q}(\fr{s}, \preRotPoly_{\inter})$,
  where 
  \begin{align*}
  \bound_{F} \leq (\frac{1}{3}) \cdot \deg \cdot \ell_{\boot} \cdot  \basis_{\boot}^2  \cdot \big(\bound_{\BR} +  \bound_{P}\big)
  \end{align*} 
 with $\bound_{P}$ being the bound on the error induced by the LWE to RLWE key switching algorithm.

From correctness of blind rotation we have that
$\acc_{\BR, F}$ $\in$ $\RLWE_{\bound_{\BR, F}, \deg, Q}(\fr{s}$,  $\preRotPoly_{\inter} \cdot X^{\tilde{\msg} \mod 2\deg})$.
Hence, from correctness of sample extraction and given that $\Delta_{Q, t'} \cdot f(\msg)$ $=$ $\coefs(\preRotPoly_{\inter} \cdot X^{\tilde{\msg} \mod 2\deg})[1]$,
we have $\vec{c}_{Q}$ $\in$ $\LWE_{\bound_{\vec{c}, Q}, \deg, Q}(\vec{s}_{F}$, $\Delta_{Q, t'} \cdot f(\msg))$,
where $\bound_{\vec{c}, Q} \leq \bound_{F} + \bound_{\BR}$.

From correctness of key switching
we have that $\vec{c}_{Q, \keySwitchKey} \in \LWE_{\bound, n, q}(\vec{s}, \Delta_{Q, t'} \cdot f(\msg))$,
where 
 $\bound_{\vec{c}, Q, \keySwitchKey} \leq \bound_{\vec{c}, Q} + \bound_{KS}$, where $\bound_{KS}$ is the bound
 induced by the key switching procedure.
From correctness of modulus switching we have that 
$\ct_{\out} \in \LWE_{\bound_{\out}, n, q}(\vec{s}, \Delta_{q, t'} \cdot f(\msg))$,
where  
\begin{align*}
\bound_{\out} \leq& (\frac{q^2}{Q^2}) \cdot (\bound_{F} + \bound_{\BR} + \bound_{KS}) \\
 &+ \frac{2}{3} + \frac{2}{3} \cdot \hamming(\vec{s}_{F}) \cdot (\norm{\Var(\vec{s}_{F})}_{\infty} + \norm{\E(\vec{s}_{F})}_{\infty}^2). 
\end{align*}

Recall that the errors induced
by the blind rotation is $\bound_{\BR} \leq (\frac{2}{3}) \cdot n \cdot u \cdot  \deg \cdot \ell_{\RGSW} \cdot \basis_{\RGSW}^2 \cdot \bound_{\brKey}$.
Similarly, recall that the error induced by the LWE to RLWE algorithm is 
$\bound_{P} \leq (\frac{1}{3}) \cdot \deg  \cdot \ell_{\packKey} \cdot \basis_{\packKey}^2 \cdot \bound_{\packKey}$.
For LWE to LWE key switching the error is $\bound_{KS} \leq (\frac{1}{3}) \cdot \deg \cdot \ell_{\keySwitchKey} \cdot \basis_{\keySwitchKey}^2 \cdot \bound_{\keySwitchKey}$.
\end{proof}

\subsection{Construction of the Rotation Polynomial}\label{sec:setting-rotation-poly}
 
In Figure~\ref{figure:functional-bootstrapping} we describe how to build the polynomials for our bootstrapping to compute arbitrary functions, and by Theorem~\ref{thm:rotation-polynomial-correctness}, we show the correctness of our methodology. 
We want to evaluate $\BootsFunction: \ZZ_t \mapsto \ZZ_Q$. 
Informally, we partition the rotation polynomials into chunks where each chunk represents a value in the domain of $\BootsFunction$. 
The chunk size is determined by the error bound of the bootstrapped LWE sample.
 Note that $\BootsFunction$ already assumes that the
domain is scaled to fit the ciphertext modulus. That is, if we want to compute $f:\ZZ_t \mapsto \ZZ_{t'}$ for some $t' \in \NN$,
then we set the lookup table $\BootsFunction[x] = \Delta_{Q, t'} \cdot f(x)$ for all $x \in \ZZ_{t}$.


\begin{theorem}\label{thm:rotation-polynomial-correctness}
Let $f:\ZZ_t \mapsto \ZZ_Q$ be a function with a lookup table $\BootsFunction$ such 
that for all $x \in \ZZ_t$, we have that $\BootsFunction[x+1] = f(x)$.
Let $(\preRotPoly_0, \preRotPoly_1) \exec \ConstPoly(\BootsFunction, \deg)$ for $\deg \in \NN$.
Let $y = \round{\frac{2 \cdot \deg}{t}} \cdot \msg + e$, for $\msg \in \ZZ_t$ and $e \in [-\roundUp{\frac{\deg}{t}}, \roundUp{\frac{\deg}{t}}]$.
Then
\begin{itemize}
\item $\coefs(\preRotPoly_0 \cdot X^y)[1] = f(\msg)$ for $y \in [0,\deg-1]$, and
\item $\coefs(\preRotPoly_1 \cdot X^y)[1] = f(\msg)$ for $y \in [\deg,2 \cdot \deg-1]$.
\end{itemize}

\end{theorem}

\begin{proof}[Correctness of the Polynomial Construction]
Let us first consider the case $y = \round{\frac{2 \cdot \deg}{t}} \cdot \msg + e \in [0,\deg-1]$.
In this case we  have
\begin{align*}
\preRotPoly_0 \cdot X^{y} = &- \sum_{i=1}^{y} \vec{p}_0[\deg - y + i] \cdot X^{i-1} \\
   &+  \sum_{i=y+1}^{\deg} \vec{p}_0[i-y] \cdot X^{i - 1} \in \R_{\deg, Q}
\end{align*}
Hence, for $y=0$, we have $\coefs(\preRotPoly_0 \cdot X^{y})[1] = \vec{p}_0[1] = \BootsFunction[\msg+1] = f(\msg)$,
and $\coefs(\preRotPoly_0 \cdot X^{y})[1] = - \vec{p}_0[\deg - y + 1] = -(-\BootsFunction[\msg+1]) = f(\msg)$ for $y \in [1,\deg-1]$.

The case $y \in [\deg,2 \cdot \deg-1]$ is analogical.
Let us first denote $y' + \deg = y$.
In particular, we have
\begin{align*}
\preRotPoly_1 \cdot X^{y} &= \preRotPoly_1 \cdot X^{y' + \deg} = - \preRotPoly_1 \cdot X^{y'} \\
&=  -\big( - \sum_{i=1}^{y'} \vec{p}_1[\deg - y' + i] \cdot X^{i-1} \\
 &~~~~~~~+ \sum_{i=y'+1}^{\deg} \vec{p}_1[i-y'] \cdot X^{i - 1} \big) \in \R_{\deg, Q}
\end{align*}
Hence, for $y = \deg$, we have $y'=0$ and $\coefs(\preRotPoly_1 \cdot X^{y})[1] = - \vec{p}_1[1] = - (-\BootsFunction[\msg+1]) = f(\msg)$.
For   $y \in (\deg,2 \cdot \deg-1]$ we have $y' \in (0, \deg -1]$
and $\coefs(\preRotPoly_1 \cdot X^{y})[1] =  \vec{p}_1[\deg - y' + 1] = \BootsFunction[\msg+1] = f(\msg)$.
\end{proof}

\subsection{Application of Functional Bootstrapping}\label{sec:applications}


Below we discuss a few simple tricks on how to 
use our bootstrapping algorithm efficiently.

 

\textbf{Computing in $\ZZ_t$.} Note that with the full domain functional bootstrapping algorithm we can correctly compute
affine functions without the need to bootstrap immediately.
Specifically, we use the linear homomorphism of LWE samples (see Lemma~\ref{lemma:linear-homo-glwe})
to compute on the ciphertext $\vec{c}_Q$ at Step 13 of $\Bootstrap$ at Figure~\ref{figure:functional-bootstrapping}.
In Section~\ref{performance} we evaluate correctness for this particular setting as in this case
computing 
the linear homomorphism will not increase the error that steams from the key switching key.

To compute the product $x \cdot y$, we compute
 $\big(\frac{(x+y)}{2}\big)^2 - \big(\frac{(x-y)}{2}\big)^2$. 
Therefore evaluating the product requires only two functional bootstrapping invocations, with compute the square, and the addition/subtraction induces only a small error when done outside the bootstrapping procedure.
Similarly, we can compute $\max(x, y)$ by as $\max(x - y, 0) + y$, which costs only one functional bootstrapping. 

\textbf{Multiple Univariate Functions With the same Input.} For efficiency it is important to note, that we might amortize the cost of computing multiple functions
on a encrypted ciphertext. After the first bootstrapping operation, we may reuse the
$\acc_{\vec{c}, i}$ values to compute the next function on the same input. 

\textbf{Extending the Arithmetic.} Finally, a powerful tool to extend the size of the arithmetic we can leverage the Chinese Reminder theorem and the fact that the ring $\ZZ_t$ for $t = \prod_{i=1}^m t_i$ where the $t_i$ are pairwise co-prime is isomorphic to
 the product ring $\ZZ_{t_1} \times \dots \times \ZZ_{t_m}$. This way we can handle modular arithmetic for a large modulus, while having cheap leveled linear operations, and one modular multiplication per factor. 
Note that we can compute arbitrary polynomials in $\ZZ_t$ in the CRT representation.
To set the rotation polynomials for each component of the ciphertext vector we do as follows.
For each number $x \in \ZZ_t$, we compute its CRT representation $x_1, \dots, x_m$.
Then we compute the polynomial $y = P(x)$ and $y$'s CRT representation $[y_i]_{i=1}^m$. 
We set the  rotation polynomials for the $i$th bootstrapping algorithm to map $x_i$ to $y_i$,
as described in Section~\ref{sec:setting-rotation-poly}. 
Furthermore, additions and scalar multiplications can still be performed without bootstrapping.

\section{Parameters, Performance, and Tests}\label{performance}

\begin{table*}[t] 
	\centering 

	\begin{tabular}{| c | c | c | c | c | c | c | c | c | c |}
		\hline
 Set  & $\FDFB$:80:6 & $\FDFB$:100:6   & $\FDFB$:80:7 & $\FDFB$:100:7 & $\FDFB$:80:8 & $\FDFB$:100:8 & $\TFHE$:100:7 & $\TFHE$:80:2 & $\TFHE$:100:2\\
		\hline  	
 $n$ &  $700$ & $1050$  &  $700$  &  $1100$  &   $700$  &  $1100$ &  $1500$  &  $424$ & $525$\\ 
 $\deg$ & $2^{11}$  & $2^{11}$     & $2^{12}$ & $2^{12}$     & $2^{13}$  & $2^{13}$ & $2^{12}$  & $2^{10}$ & $2^{10}$ \\ 
 $q$ &  $2^{12}$ & $2^{12}$    &  $2^{13}$ &   $2^{13}$ & $2^{14}$ &   $2^{14}$ & $2^{13}$ &  $2^{11}$  & $2^{11}$  \\
 $Q$ & $Q_1$  & $Q_1$    &  $Q_2$   & $Q_2$  & $Q_2$   & $Q_2$  &  $Q_3$ & $Q_4$ & $Q_4$ \\ 
 $\basis_{\boot}$ & $2^{11}$  & $2^{11}$     & $2^{11}$   &  $2^{11}$ & $2^{8}$   &  $2^{8}$ & N/A & N/A &  N/A  \\
 $\basis_{\RGSW}$ &  $2^{11}$ & $2^{11}$    & $2^{9}$  &  $2^{9}$ &  $2^{9}$  &  $2^{9}$ & $2^{16}$  &  $2^{11}$ & $2^{11}$  \\ 
 $\basis_{\keySwitchKey}$ &   $2^{6}$  & $2$  &  $2^{4}$   & $2$ &  $2^{4}$   & $2$ &  $2$   & $2^{4}$ & $2^{4}$  \\
 $\basis_{\packKey}$ &  $2^{13}$  & $2^{13}$  &  $2^{13}$ & $2^{13}$ &   $2^{13}$ & $2^{13}$ &  N/A & N/A &  N/A  \\ 
 $\stddev$ & $2^{38}, 3.2$ & $2^{41}, 3.2$  &  $2^{39}, 3.2$  & $2^{41}, 3.2$ & $2^{39}, 3.2$  & $2^{41}, 3.2$ & $2^{18}, 3.2$ & $2^{16}, 3.2$ & $2^{16}, 3.2$  \\
  $\vec{s}$ & $[0, 1], 64$ & $[0, 1], 64$  &  $[0, 1], 64$   &  $[0, 1], 64$  &  $[0, 1], 64$ &  $[0, 1], 64$  &  $[0, 1], 64$ & $[0, 1], -$ & $[0, 1], -$  \\
 $\fr{s}$ & $\ZZ_Q^{\deg}, -$ & $\ZZ_Q^{\deg}, -$  &  $\ZZ_Q^{\deg}, -$  & $\ZZ_Q^{\deg}, -$ &  $\ZZ_Q^{\deg}, -$ & $\ZZ_Q^{\deg}, -$ &  $\ZZ_Q^{\deg}, -$ &  $\ZZ_Q^{\deg}, -$ & $\ZZ_Q^{\deg}, -$  \\
 \hline
\end{tabular}  
\caption{Parameter Sets. The modulus $Q_1 = 2^{62} - 65535$, $Q_2 = 2^{63} - 278527$, $Q_3 = 2^{48} - 163839$, and  $Q_4 = 2^{32} - 139263$.
 For all parameters sets we set the error variance of fresh (R)LWE samples to $10.24$ which corresponds to standard deviation $3.2$. The key $\vec{s}$ is always binary; thus $u=1$.}\label{tbl:parameter-sets}

\end{table*}

\begin{table*}
 \centering 
 \begin{tabular}{| c | c | c | c | c | c | c | c |}
 \hline
  \diagbox{Set}{$\log_2(t)$}   & 6 & 7 &  8 &  9 & 10 &  11  &  Key \\  
 \hline   
$\FDFB$:80:6 &  $2^{-30}$, $2^{-23}$  &   $2^{-9}$, $0.05$ & $0.12$, $0.57$ & $0.44$, $0.88$ & $0.70$, $0.96$ & $0.84$, $0.99$ & 93\% \\ 
$\FDFB$:100:6 &   $2^{-32}$, $2^{-21}$  &   $2^{-10}$, $0.09$ & $0.11$, $0.63$ &  $0.43$, $0.90$  & $0.69$, $0.97$ & $0.84$, $0.99$ &   97\% \\  
$\FDFB$:80:7 &    $0.0$, $0.0$  &   $2^{-31}$, $2^{-21}$ & $2^{-9}$, $0.07$ &  $0.12$, $0.61$  & $0.44$, $0.89$ &  $0.69$, $0.97$ &  95\% \\ 
$\FDFB$:100:7 &    $0.0$, $0.0$  &  $2^{-31}$, $2^{-18}$ & $2^{-9}$, $0.13$ &  $0.12$, $0.67$  & $0.44$, $0.91$  &  $0.69$, $0.97$ &  95\% \\ 
$\FDFB$:80:8 &    $0.0$, $0.0$  &   $0.0$, $0.0$ & $2^{-26}$, $2^{-17}$ &  $2^{-8}$, $0.12$  & $0.16$, $0.66$ &  $0.48$, $0.90$ &  67\% \\ 
$\FDFB$:100:8 &    $0.0$, $0.0$  &  $0.0$, $0.0$ & $2^{-26}$, $2^{-13}$ &  $2^{-8}$, $0.23$   & $0.16$, $0.75$  &  $0.48$, $0.93$ &  67\% \\ 
$\TFHE$:100:7 &    $0.0$  & $2^{-24}$ & $2^{-7}$ &  $0.18$  & $0.50$ & $0.73$ &   99\% \\ 
$\TFHE$:80:2 &    $0.57$  & $0.77$ & $0.88$ &  $0.94$  & $0.97$ &  $0.98$ &  14\% \\ 
$\TFHE$:100:2 &    $0.58$  &  $0.78$ & $0.89$ &  $0.94$  & $0.97$  & $0.98$ &  17\% \\ 
\hline
\end{tabular} 
\caption{Upper bound on then probability to have an error in the outcome of the bootstrapping algorithms.  For our full domain bootstrapping, we give two values. The first is the upper bound on the error probability for ciphertexts after bootstrapping. The second value is the upper bound for a ciphertext after bootstrapping and computing an affine function of size 784. We emphasize that the probabilities are upper-bounds, and the empirical correctness errors are much lower.}\label{tbl:parameter-sets-correctness}
\end{table*}

\begin{table*}
	\centering 
	\begin{tabular}{|c | c | c | c | c | c | c | c | c | c | c | c |}
		\hline
	 \multirow{2}{*}{Set} & \multirow{2}{*}{Time [ms]} & \multirow{2}{*}{PK [MB]} & \multirow{2}{*}{CT [KB]} & \multicolumn{3}{c|}{RLWE} & \multicolumn{4}{c|}{LWE}  \\
		& & & & uSVP & dec & dual  &  uSVP & dec & dual  &    Hyb. Primal \\ 
		\hline   
$\FDFB$:80:6 & 3655 &  737.13 & 1.40 & 129.5 & 183.2 & 135.7 & 87.8 & 303.8 &  82.7 & 81.6 \\ 
$\FDFB$:100:6 & 4838 &  1816.03 & 2.10  & 129.5 &  183.2 & 135.7 & 138.0 &  668.9 & 127.6 &  100.1  \\  
$\FDFB$:80:7 & 9278 &  2351.95 & 1.40 & 282.4 & 306.1 &  294.7  & 87.8 & 309.5 &  82.7 & 81.6 \\ 
$\FDFB$:100:7 & 13781 &  4624.31  & 2.20 & 282.4 & 306.1 &  294.7 &  138.2 &  644.3 & 129.0 &  100.0  \\ 
$\FDFB$:80:8 & 33117 &  7388 & 1.40 & 548.5 &  596.2 &  570.1  & 87.8 & 309.5 & 82.7 & 81.6 \\ 
$\FDFB$:100:8 & 45266 &  14318  & 3.00 & 548.5 &  596.2 &  570.1  & 138.2 & 644.3 & 129.0 & 100.0  \\ 
$\TFHE$:100:7 & 7757 &  921.91  & 2.43 & 236.9 & 255.5 &  247.4 &  139.6 &  322.7 & 128.9 & 100.8  \\ 
$\TFHE$:80:2 & 306 &  31.72  & 0.72 &  116.9  & 128.3 & 125.3 &  81.3 &  207.3 & 86.2 & 88.9   \\ 
$\TFHE$:100:2 & 321 &  25.62 & 0.58 & 116.9 &  128.3 & 125.3 &  100.0  & 264.3 & 107.1 &   111.0  \\ 
\hline
\end{tabular} 
 \caption{Complexity and security for the parameter sets.  
 We computed the size of the keys and ciphertexts by taking the byte size of every field element and counting the number of field elements.}\label{tbl:parameter-sets-size-and-security}
\end{table*}

\begin{table*}
	\centering
	\begin{tabular}{| c || c | c | c | c | c | c | c | c | c | c | c | c|}
		\hline
		    & \multicolumn{6}{c}{Evaluation Time MNIST-1-$f$} & \multicolumn{6}{c|}{Accuracy  MNIST-1-$f$} \\
	 \diagbox{Param.}{$\log_2(t)$}  & 6 & 7 & 8 & 9 & 10 & 11 & 6 & 7 & 8 & 9 & 10 & 11 \\
	  \hline
$\FDFB$:80:6  &  \multicolumn{6}{c|}{0.027} & 0.80 & 0.95 & 0.95 & 0.95 & 0.80 & 0.60 \\
$\FDFB$:100:6  &  \multicolumn{6}{c|}{0.11} & 0.75 & 0.90 & 0.90 & 0.90 & 0.80 & 0.75  \\
$\FDFB$:80:7 &  \multicolumn{6}{c|}{0.08} & 0.90 & 0.90 & 0.95 & 0.95 & 0.90 & 0.70 \\
$\FDFB$:100:7 &  \multicolumn{6}{c|}{0.35} & 0.90 & 0.90 & 0.95 & 0.90 & 0.90 & 0.80\\
$\FDFB$:80:8 &  \multicolumn{6}{c|}{0.17} & 0.85 & 0.90 & 0.95 & 0.90 & 0.95 & 0.90\\
$\FDFB$:100:8 &  \multicolumn{6}{c|}{0.76} & 0.80 & 0.95 & 0.90 & 0.95 & 0.95 & 0.90 \\
$\TFHE$:80:2 & 99.24 & 135.88 & 177.28 & 225.84 & 289.82 & 392.14  & 0.82 & 0.94 & 0.95 & 0.95 & 0.95 & 0.95 \\
$\TFHE$:100:2 & 122.18 & 165.9 & 215.44 & 274.44 & 351.42 & 460.94 & 0.82 & 0.94 & 0.95 & 0.95 & 0.95 & 0.95 \\
		\hline 
		    & \multicolumn{6}{c}{Evaluation Time  MNIST-2-$f$} & \multicolumn{6}{c|}{Accuracy  MNIST-2-$f$} \\ 
	  \hline
$\FDFB$:80:6  &  \multicolumn{6}{c|}{0.14} & 0.10 & 0.20 & 0.45 & 0.85 & 0.90 & 0.85 \\
$\FDFB$:100:6  &  \multicolumn{6}{c|}{0.23} & 0.10 & 0.20 & 0.70 & 0.95 & 0.95 & 0.90  \\
$\FDFB$:80:7 &  \multicolumn{6}{c|}{0.37} & 0.15 & 0.25 & 0.85 & 0.90 & 0.90 & 0.90 \\
$\FDFB$:100:7 &  \multicolumn{6}{c|}{0.65} & 0.15 & 0.30 & 0.85 & 0.95 & 0.90 & 0.85\\
$\FDFB$:80:8 &  \multicolumn{6}{c|}{1.1} & 0.15 & 0.30 & 0.85 & 0.90 & 0.95 & 0.90\\
$\FDFB$:100:8 &  \multicolumn{6}{c|}{1.8} & 0.15 & 0.30 & 0.90 & 0.90 & 0.95 & 0.95 \\
$\TFHE$:80:2 & 24.56 & 33.64 & 43.88 & 55.9 & 71.74 & 97.08  & 0.14 & 0.22 & 0.57 & 0.93 & 0.93 & 0.93 \\
$\TFHE$:100:2 & 33.64 & 41.06 & 53.3 & 67.94 & 86.98 & 114.1 & 0.14 & 0.22 & 0.57 & 0.93 & 0.93 & 0.93  \\
		\hline
	\end{tabular}
	\caption{Time to evaluate the MNIST-1-$f$ and MNIST-2-$f$ neural networks.  
	 The time to evaluate the neural networks is given in hours. 
	 Due to extremely high execution times, we only present estimates for $\TFHE$, which we base on Table~\ref{tbl:app-add-scal}.
	 For accuracy we show values for $f = \ReLU$, but stress that we can use arbitrary univariate functions without impacting execution
	 time. We discretized the models using a discretization factor $\delta = 6$ for MNIST-1-$f$ and $\delta = 4$ for MNIST-2-$f$.}\label{tbl:mnist-eval}.
\end{table*}

\begin{table}
	\centering
	\begin{tabular}{| c | c | c |}
		\hline
		& MNIST-1-$f$ & MNIST-2-$f$ \\ 
		\hline
		Input & $[x_i]_{i=1}^{784}$ & $[x_i]_{h=1i=1,j=1}^{1,28,28}$ \\
		1 & $\Dense$:784:$f$ & $\ConvTwoD$:4:(6,6):$f$ \\
		2 & $\Dense$:510:$f$ & $\AvgPoolTwoD$:(2,2) \\
		3 & $\Softmax$:10 & $\ConvTwoD$:16:(6,6):$f$ \\
		4 &  & $\AvgPoolTwoD$:(2,2) \\
		5 & & $\AvgPoolTwoD$:16:(3,3):$f$ \\
		6 & & $\Dense$:64:$f$ \\
		7 & &  $\Softmax$:10 \\
		\hline 
		Baseline accuracy for & \multirow{2}{*}{0.96} & \multirow{2}{*}{0.98} \\
		$\ReLU(x) = \max(0, x)$& & \\
		\hline
	\end{tabular}
	\caption{Specification for both models that we train and evaluate and the baseline accuracy of the non-discretised model using floating point numbers. }\label{tbl:mnist-models}.
\end{table}  
\begin{table} 
	\centering
	\begin{tabular}{| c |  c | c |}
		\hline
		Param. & Addition & Scalar Multiplication \\
		\hline
		$\FDFB$:80:6 $[\mu s]$ &    1 &  5\\
		$\FDFB$:100:6 $[\mu s]$ &   1 &  6\\
		$\TFHE$:80:2 $[s]$  &   13.2 & 13.2\\
		$\TFHE$:100:2 $[s]$ &   16.2 & 16.2 \\
		\hline
		$\FDFB$:80:7 $[\mu s]$ &   1 & 4\\
		$\FDFB$:100:7 $[\mu s]$ &  1 & 7 \\
		$\TFHE$:80:2 $[s]$ &  15.6 & 19.2 \\
		$\TFHE$:100:2 $[s]$ &  19.2 & 19.2 \\
		\hline
		$\FDFB$:80:8 $[\mu s]$ &  1 & 3 \\
		$\FDFB$:100:8 $[\mu s]$ &  1 & 4 \\
		$\TFHE$:80:2 $[s]$ &   18 & 18\\
		$\TFHE$:100:2 $[s]$ &  21 & 21 \\
		\hline 
	\end{tabular}
	\caption{Addition and scalar multiplication timings. Timings in \textbf{microseconds} for $\FDFB$ parameter sets. Timings in \textbf{seconds} for $\TFHE$ sets. We consider the timings for operations from $\ZZ_{2^k}\times\ZZ_{2^k}$ to $\ZZ_{2^k}$ for $k \in \{6,7,8\}$. }\label{tbl:app-add-scal}
\end{table}
 
First, we note that if the ratio $Q/q$ is high, then the error after modulus switching from $Q$ to $q$ is dominated by the Hamming weight of the LWE secret key $\vec{s}$. To minimize the error, we choose $q = 2\cdot \deg$ to avoid another modulus switching.
To be able to compute number theoretic transforms we choose $Q$ prime and $Q = 1 \mod 2 \cdot \deg$.
An important observation is that the error does not depend on the RLWE secret key $\fr{s}$.
Hence, we choose $\fr{s}$ from the uniform distribution over $\R_{\deg, Q}$.
For all RLWE samples, we set the standard deviation of the error as $3.2$. 
We choose the LWE secret key $\vec{s}$ from $\{0, 1\}$ such that the Hamming weight of the vector is $64$.
Recall that LWE with binary keys is asymptotically as secure as the standard LWE \cite{TOC:Micciancio18}, and is also used in many
implementations of TFHE.
The Hamming weight is the same as used in the HElib library \cite{C:HalSho14,EC:HalSho15,EC:Albrecht17}. 
To compensate for the loss of security due to a small hamming weight, we choose the standard deviation of LWE samples error much higher
that the in case of RLWE. 
For binary components of $\vec{s}$ we have $\norm{\vec{u}} = 1$ which is as in TFHE \cite{AC:CGGI16,JC:CGGI20}.
We choose the decomposition bases $\basis_{\boot}$, $\basis_{\RGSW}$, and $\basis_{\packKey}$ to minimize the amount of computation.
Due to the higher noise variance of the LWE samples in the key switching key, we choose the decomposition basis  $\basis_{\keySwitchKey}$ 
relatively low, but we note that in terms of efficiency, key switching does not contribute much.

To evaluate security, we use the LWE estimator 
\cite{albrecht2015concrete} commit \texttt{fb7deba} and the Sparse LWE estimator 
 \cite{EPRINT:CheSonYhe21}  commit \texttt{6ab7a6e}  based on the LWE Estimator. 
The LWE Estimator estimates the security level of a given parameter set by calculating the complexity of the meet in the middle
exhaustive search algorithm, Coded-BKW \cite{C:GuoJohSta15}, dual-lattice attack and small/sparse secret variant \cite{EC:Albrecht17}, lattice-reduction + enumeration \cite{RSA:LinPei11}, primal attack via uSVP \cite{ICISC:AlbFitGop13,ACISP:BaiGal14}
and Arora-Ge algorithm \cite{ICALP:AroGe11} using Gr\"obner bases \cite{EPRINT:ACFP14}.
The sparse LWE estimator is build on top of \cite{albrecht2015concrete}, but additionally 
considers a variant \cite{EPRINT:SonChe19} of Howgrave-Graham's hybrid attack \cite{C:HowgraveGraham07}.

We calculate the correctness by taking the standard deviation of the
bootstrapped sample. 
In particular, given that $\bound$ is the bound on the variance of the error we set $\stddev_{\bound} = \sqrt{\bound}$, and compute the probability of a correct bootstrapping as $\erf(\frac{q}{2 \cdot t \cdot \stddev_{\bound} \cdot \sqrt{2}})$.
In Table~\ref{tbl:parameter-sets-correctness} for the sake of readability, we approximate the probability of getting an incorrect output as a power of 2.
That is we find the smallest $i \in \NN$ such that $2^{-i} \leq 1 - \erf(\frac{q}{2 \cdot t \cdot \stddev_{\bound} \cdot \sqrt{2}})$.
For sets 1 to 4 we show correctness calculated from the bound $\bound_{\deg}$ as given in Theorem~\ref{thm:correctness-of-functional-bootstrapping}. Furthermore, we show correctness when computing affine functions after bootstrapping.
 For TFHE we take the error bound
$\bound_{\deg}$ as given by Theorem~\ref{thm:correctness-tfhe-bootstrapping}. 
In both cases, according to our parameter selection, we have that $\bound_{\deg}$ is computed from
$\bound_{\ct}$ since it is higher than $\bound_{\fresh}$. Note that $\bound_{\ct}$ is a ciphertext
output by a previous bootstrapping.

\textit{Parameter Sets.}  
Table~\ref{tbl:parameter-sets} shows our different parameter sets, Table~\ref{tbl:parameter-sets-correctness}
gives correctness for each parameter set and the size of the message space, and
Table~\ref{tbl:parameter-sets-size-and-security} contains an overview of the ciphertext size, estimated security level and time to run single bootstrapping operation.
      
We denote the parameter sets $\FDFB$:$s$:$t$ and $\TFHE$:$s$:$t$ where $s$ is the security level,
and $t$ is the number of bits of the plaintext space for which the probability of having an incorrect outcome is very low.
For example, $\FDFB$:$80$:$6$ offers $80$ bits of security and is targeted at $6$-bit plaintexts.
We stress that the targeted plaintext space does not exclude larger plaintext moduli.
In particular, according to Table~\ref{tbl:parameter-sets-correctness}, all parameter sets can handle larger plaintexts, but the probability that the least significant bit of the input message changes obviously grows with the plaintext space size, shifting the scheme from FHE towards homomorphic encryption for approximate arithmetic.

\textit{Implementation.} 
We ran our experiments on a virtual machine with 60 cores, featuring 240GB of ram and using Intel Xeon Cascade lake processors, and the measurements are averaged over five runs.
We implemented our bootstrapping technique by modifying the latest\footnote{\texttt{commit} 612618aca7ca2c5df5b06a0f6c32f1db680722964} version of the PALISADE 
\cite{PALISADE}
and using the Intel HEXL \cite{HEXL} 
already featured an implementation of \cite{EPRINT:DucMic14b} and \cite{AC:CGGI16}. 
Note that the timings in PALISADE  for TFHE are different than reported in \cite{AC:CGGI16}. 
This is mostly due to an optimization in TFHE that chooses the modulus $Q$ as $Q=2^{32}$ or $Q=2^{64}$ and uses
the natural modulo reduction of unsigned int and long C++ types.
Then instead of computing the number theoretic transform, TFHE computes the fast Fourier transform. 
We note that the same optimization is possible for our algorithms. 

\textit{Evaluation.}  
Table~\ref{tbl:app-add-scal} shows the efficiency of modular addition and scalar multiplication targeting different plaintext spaces, respectively, for our parameter sets. We compare our approach to operations on binary ciphertexts. To this end, we use the vertical look-up tables from \cite{AC:CGGI17}, as they turned out to be faster
than a circuit-based implementation. Note that, while vertical LUTs are significantly faster than the horizontal tables, which were also introduced
in \cite{AC:CGGI17}, they can compute at most one function at the same time. Therefore multiple tables need to be maintained for each output bit, 
requiring additional memory. 
  
To evaluate a neural network on encrypted inputs, we discretized two neural networks that recognize handwritten digits from the MNIST dataset.
Note that the discretization of a floating-point neural network to integers models with modular arithmetic is a delicate operation. Many operations such as normalization ca not be applied in the same manner, and therefore the discretization process often results
in a drop of accuracy. 
The configurations and accuracies are given at Table~\ref{tbl:mnist-models}.
In what follows, we give the notation used to specify the neural networks.
We always denote the weights with $w$'s and biases with $b$'s.
 By $\Dense$:$d$:$f$ we denote a dense layer that on input a vector $[x_i]_{i=1}^{m}$
outputs a vector $\big[f(b_j + \sum_{i=1}^{m} x_i \cdot w_{j, i})\big]_{j=1}^{d}$. 

Convolution layers are denoted as $\ConvTwoD$:$h$:$(d, d)$:$f$, where $h$ is the number of filters, $(d, d)$ specify the shape of the filters and $f$ is an activation function. A $\ConvTwoD$ layer takes as input a tensor $[x_{g,i,j}]_{g=1,i=1,j=1}^{c,m,l}$ and
outputs 
\begin{align*}
\bigg[\Big[f\big(b_{i', j'}^{(k)} + \sum_{g=1,i=1,j=1}^{c, d, d} x_{g, i+i', j+j'} \cdot w^{(k)}_{g, i, j}\big)\Big]_{i'=0,j'=0}^{m-d,l-d}\bigg]_{k=1}^{h}.
\end{align*}
Average pooling $\AvgPoolTwoD$:$(d, d)$ on input a tensor  $[x_{g,i,j}]_{g=1,i=1,j=1}^{c,m,l}$
outputs 
\begin{align*}
\bigg[ \Big[\frac{\sum_{i=1,j=1}^{d, d} x_{g, i+i', j+j'}}{d^2}\Big]_{i'=0,j'=0}^{m-d,l-d} \bigg]_{g=1}^c
\end{align*}
Finally, $\Softmax$:$d$ takes as input a vector $[x_i]_{i=1}^{m}$, computes $\vec{z} = \big[b_j + \sum_{i=1}^{m} x_i \cdot w_{j, i}\big]_{j=1}^{d}$ and outputs $\big[\Softmax(\vec{z}, i) = \frac{e^{\vec{z}[i]}}{\sum_j^d e^{\vec{z}[j]}}\big]_{i=1}$.

We compare our FHE with our bootstrapping with
a functionally equivalent implementation based on vertical LUTs \cite{AC:CGGI17}.
In particular, for activation functions, we run $k$, $k$-to-1 bit LUTs for $k \in \{6,7,8,9,10,11\}$ and 
to compute addition and scalar multiplication for the affine function, we must use $k$, $2k$-to-1 bit LUTs.
We note that we also tested the method of Carpov et al. \cite{RSA:CarIzaMol19} to compute the activation function on set $\TFHE$:100:7.
However, the timing to evaluate the affine function in the LUT-based method is so overwhelming that we did not notice any significant difference. 
We include the evaluation accuracies for plaintext spaces that exceed the target space. Such settings are extremely useful for computations in which a small amount of error is acceptable, such as activation functions in neural network inference.  
We can see that our approach significantly outperforms LUT-based techniques.  The most important reason for this is that we can compute affine functions within neurons almost instantly due to the linear homomorphism of LWE samples. Binary ciphertext requires the evaluation of expensive LUT at each step. In the case of the LUT-based approach, we note that the computation will be exact. 
Therefore, the accuracy will be identical to the discretized neural network with applied modulo reduction.
Table~\ref{tbl:mnist-eval} shows the time required to evaluate the whole neural network. 
We test our
network on 20 samples for each plaintext space and achieve the best speedups for sets $\FDFB$:80:6 and $\FDFB$:100:6 that outperform the LUT-based technique by a factor of over 3600 and 1000. Even in the worst case with the MNIST-2-$f$ network, for which the affine functions are significantly smaller, we obtain speedups 44.4 and 29.6  for $\FDFB$:80:8 against $\TFHE$:80:2, and respectively $\FDFB$:100:8 against $\TFHE$:100:2 for 8-bit plaintexts.

By computing the activation functions of the network over 60 cores, we manage to lower the evaluation significantly. 
To estimate the speedup one would gain by parallelizing the LUT-based technique, we averaged the speedup gained by the $\FDFB$ and applied it to the LUT-based timings. We did not parallelize our bootstrapping technique itself. However, we note that there is a lot of room to do so.

\section{Conclusion}

We believe that our evaluations showed that it is practically feasible to evaluate neural networks over encrypted data 
using fully homomorphic encryption.
We emphasize that in our method, the choice of the activation functions does not play a role.
We showed that parallelization might play a pivotal role in further reducing the timing for oblivious neural network inference. Therefore, an open question is how much time can be improved when exploiting graphics processing units or larger clusters.
Finally, we believe that our method may
either be a complementary algorithm to CKKS methods or a competitive alternative that gives the evaluation's exact results instead of approximate.

\bibliographystyle{alpha}
\bibliography{../cryptobib/abbrev3,../cryptobib/crypto,additionalBiblio}

\appendix

\section{Specifications of Algorithms from the Preliminaries}\label{sec:algs-from-preliminaries}

\begin{definition}[External Product]
The external product $\extProd$ given as input $\ma{C} \in \GGSW_{\bound_{\ma{C}}, n, \deg, q, \basis}(\vec{s}, \msg_{\vec{C}})$ and $\vec{d} \in \GLWE_{\bound_{\vec{d}}, n, \deg, q}(\vec{s}, \msg_{\vec{d}})$, 
outputs the following 
\begin{align*}
\extProd(\ma{C}, \vec{d}) = \Decomp_{\basis, q}(\vec{d}^{\top}) \cdot \ma{C}.
\end{align*} 
\end{definition}

\begin{figure}
\center
\fbox{
\procedure[syntaxhighlight=auto ]{$\CMux(\ma{C}, \vec{g}, \vec{h})$}{  
    \textbf{Input:~}  
    \text{Takes as input
	$\ma{C} \in \GGSW_{\bound_{\ma{C}}, n, \deg, q, \basis}(\vec{s}, \msg_{\ma{C}})$,} \\
	\text{where $\msg_{\ma{C}} \in \BB$,  $\vec{g} \in \GLWE_{\bound_{\vec{g}}, n, \deg, q}(\vec{s}, .)$} \\
	\text{and $\vec{h} \in \GLWE_{\bound_{\vec{h}}, n, \deg, q}(\vec{s}, .)$.}
 \\[][\hline]  
\pcln \text{Compute $\vec{d} \exec \vec{g} - \vec{h}$.} \\
\pcln \text{Compute $\vec{c}_{\out} \exec \extProd(\ma{C}, \vec{d})^{\top} + \vec{h}$.} \\
\pcln \text{Output $\vec{c}_{\out}$.}
}   
} 
\caption{CMux gate.}
\label{figure:cmux}
\end{figure}


\begin{figure}
\center
\fbox{
\procedure[syntaxhighlight=auto ]{$\ModSwitch(\vec{c}, Q, q)$}{  
    \textbf{Input:~}  
    \text{Takes as input $\vec{c} = [b, \vec{a}^{\top}]^{\top} \in \LWE_{\bound, n, Q}(\vec{s}, .)$,}\\
    \text{$Q \in \NN$ and $q \in \NN$.}
 \\[][\hline]  
\pcln  \text{Set $b_q \exec \round{\frac{q \cdot b}{Q}}$.} \\ 
\pcln  \text{Set $\vec{a}_q = \round{\frac{q \cdot \vec{a}}{Q}}$.} \\  
\pcln \text{Output $\ct_{\out} \exec [b_q, \vec{a}_q^{\top}]^{\top} \in \ZZ_q^{(n+1) \times 1}$.}
}   
} 
\caption{Modulus switching algorithm.}
\label{figure:mod-switch}
\end{figure}


To conveniently describe our $\SampleExtract$, we first need to introduce a step function which we call $\NegSgn:\RR \mapsto \{-1, 1\}$ and which is defined as
\begin{align*}
\NegSgn(x) = 
\begin{cases}
1 & \text{if}~ 	x \ge  0 \\
-1 & \text{otherwise}
\end{cases}
\end{align*}
Note that the function is equal in all points to the sign function except for $x=0$, which its image is $-1$.

\begin{figure}
\center
\fbox{ 
\procedure[syntaxhighlight=auto ]{$\SampleExtract(\ct, k):$}{  
    \textbf{Input:~}  
    \text{Takes as input $\ct = [\fr{b}, \fr{a}]^{\top} \in \RLWE_{\bound, \deg, q}(\fr{s}, .)$}\\
    \text{  and an index $k \in [\deg]$.}
 \\[][\hline]  
\pcln \text{Set $b' = \coefs(\fr{b})[k]$.}  \\
\pcln \text{Set  $\vec{a} \exec  \coefs(\fr{a}) \in \ZZ_q^{\deg}$. }\\
\pcln \text{For all $i \in [\deg]$}\\
\pcln\pcind \text{Set $\vec{a}'[i] \exec \NegSgn(k - i + 1) \cdot  \vec{a}[(k - i \text{~mod~} \deg) + 1]$.} \\
\pcln \text{Return $\vec{c} \exec (\vec{a}', b')$.}
}    
} 
\caption{Sample and key extraction. The $\KeyExtract(\fr{s})$ function on input key $\fr{s} \in \R_{\deg, q}$ outputs its coefficient vector. That is it outputs $\coefs(\fr{s})$.}
\label{figure:mod-switch}
\end{figure}

%

\begin{figure*}
\center
\fbox{ 
\procedure[syntaxhighlight=auto ]{$\KeySwitchSetup(\bound_{\keySwitchKey}, n, n', \deg, q, \vec{s}, \vec{s}', \basis_{\keySwitchKey})$}{  
    \textbf{Input:~}  
    \text{A bound $\bound_{\keySwitchKey} \in \NN$, dimensions $n, n' \in \NN$,} \\
    \text{ a degree $\deg \in \NN$ and a modulus $q \in \NN$ that define $\R_{\deg, q}$,} \\
      \text{vectors  $\vec{s} \in \ZZ_q^n$, $\vec{s}' \in \R_{\deg, q}^{n'}$ and a basis $\basis_{\keySwitchKey} \in \NN$.}
 \\[][\hline]  
\pcln \text{Set $\ell_{\keySwitchKey} = \roundUp{\log_{\basis_{\keySwitchKey}}{q}}$.} \\
\pcln \text{For $i \in [n]$,  $j \in [\ell_{\keySwitchKey}]$} \\  
\pcln\pcind \text{$\keySwitchKey[i,j] \exec \GLWE_{\bound_{\keySwitchKey}, n', \deg, q}(\vec{s}', \vec{s}[i] \cdot \basis_{\keySwitchKey}^{j-1})$.} \\
\pcln \text{Output $\keySwitchKey \in \GLWE_{\bound_{\keySwitchKey}, n', \deg, q}(\vec{s}', .)^{n \times \ell_{\keySwitchKey}}$.} 
}
  
\pchspace[0.2cm]

\procedure[syntaxhighlight=auto ]{$\KeySwitch(\vec{c}, \keySwitchKey)$}{  
    \textbf{Input:~}  
   \text{Takes as input  $\vec{c} = [b, \vec{a}^{\top}]^{\top} \in \LWE_{\bound_{\vec{c}}, n, q}(\vec{s}, .)$} \\
   \text{and a key switching key $\keySwitchKey \in \GLWE_{\bound_{\keySwitchKey}, n', \deg, q}(\vec{s}', .)^{n \times \ell_{\keySwitchKey}}$.}
 \\[][\hline]  
\pcln \text{Denote $\ell_{\keySwitchKey} = \roundUp{\log_{\basis_{\keySwitchKey}}{q}}$. }\\
\pcln \text{Compute $\ma{A} \exec \Decomp_{\basis_{\keySwitchKey}, q}(\vec{a}) \in \ZZ_{\basis_{\keySwitchKey}}^{n \times \ell_{\keySwitchKey}}$.    }\\
\pcln \text{Output $c_{\out} \exec [\fr{b}, 0]^{\top} - \sum_{i=1}^n \sum_{j=1}^{\ell_{\keySwitchKey}} , \ma{A}[i,j] \cdot \keySwitchKey[i,j]$. } 
}    
} 
\caption{Key switching algorithm and its setup.}
\label{figure:keyswitch}
\end{figure*}

\begin{figure*}
\center
\fbox{ 
\procedure[syntaxhighlight=auto ]{$\BRKeyGen(n, \vec{s}, \vec{u}, \bound_{\brKey}, \deg, Q, \basis_{\RGSW}, \fr{s})$}{  
    \textbf{Input:~}  
    \text{A dimension $n \in \NN$, a secret key $\vec{s} \in \ZZ_t^{n \times 1}$,} \\
    \text{a vector $\vec{u} \in \ZZ^{u}$, a bound $\bound_{\brKey}$, a degree $\deg \in \NN$, }\\
 \text{a modulus $Q$ defining the ring $\R_{\deg, Q}$, a basis $\basis_{\RGSW} \in \NN$,} \\
 \text{and a secret key $\fr{s} \in \R_{\deg, Q}$.}
 \\[][\hline]  
\pcln  \text{Set a matrix $\ma{Y} \in \BB^{n \times u}$ as follows:} \\
\pcln \text{For $i \in [n]$ do} \\
\pcln\pcind \text{Set the $i$-th row of $\ma{Y}$, to satisfy $\vec{s}[i] = \sum_{j=1}^{u} \ma{Y}[i, j] \cdot \vec{u}[j]$,} \\
\pcln \text{For $i \in [n]$, $j \in [u]$ } \\
\pcln\pcind \text{Set $\brKey[i,j] = \RGSW_{\bound_{\brKey}, \deg, Q, \basis_{\RGSW}}(\fr{s}, \ma{Y}[i,j])$.} \\
\pcln \text{Output  $\brKey$.}
}
  
\pchspace[0.1cm]

\procedure[syntaxhighlight=auto ]{$\BlindRotate(\brKey, \acc, \ct, \vec{u})$}{  
    \textbf{Input:~}  
   \text{Blind rotation key $\brKey$ $=$  $\RGSW_{\bound_{\brKey}, \deg, Q, \basis_{\RGSW}}(\fr{s}, .)^{n \times u}$,} \\
   \text{an accumulator $\acc \in \RLWE_{\bound_{\acc}, \deg, Q}(\fr{s}, .)$,} \\
   \text{a ciphertext $\ct \in  \LWE_{\bound_{\ct}, n, 2 \cdot \deg}(\vec{s}, .)$,} \\
   \text{where $\vec{s} \in \ZZ_t^{n}$ with $t \leq q$, and a vector $\vec{u} \in \ZZ^{u}$.}
 \\[][\hline]  
\pcln \text{Let $\vec{c} = [b, \vec{a}^{\top}]^{\top} \in \ZZ_{2 \cdot \deg}^{(n+1) \times 1}$.}\\
\pcln \text{For $i \in [n]$ do} \\ 
\pcln\pcind  \text{For $j \in [u]$ do}\\ 
\pcln\pcind\pcind \text{$\acc \exec \CMux(\brKey[i, j], \acc \cdot X^{-\vec{a}[i] \cdot \vec{u}[j]}, \acc)$.} \\
\pcln \text{Output $\acc$.} 
}    
} 
\caption{Blind Rotation and its setup.}
\label{figure:blind-rotation}
\end{figure*}

\begin{figure}
\center
\fbox{ 
\procedure[syntaxhighlight=auto ]{$\Bootstrap(\brKey,  \vec{u}, \ct, \preRotPoly, \keySwitchKey)$}{  
    \textbf{Input:~}  
    \text{Blind rotation key $\brKey$ $=$  $\RGSW_{\bound_{\brKey}, \deg, Q, \basis_{\brKey}}(\fr{s}$, $.)^{n \times u}$,} \\
    \text{a vector $\vec{u} \in \ZZ^{u}$,} \\
\text{a LWE sample $\ct$ $=$ $\LWE_{\bound_{\ct}, n, q}(\vec{s}$, $.)$ $=$ $[b$, $\vec{a}^{\top}]^{\top} \in \ZZ_q^{(n+1) \times 1}$, } \\
\text{a polynomial  $\preRotPoly \in \R_{\deg, Q}$,} \\
\text{and a LWE to LWE keyswitch key $\keySwitchKey$.}
 \\[][\hline]  
\pcln \text{Set $\ct_{\deg} \exec \ModSwitch(\ct, q, 2 \cdot \deg) = [b_{\deg}, \vec{a}_{\deg}^{\top}]^{\top}$.} \\
\pcln \text{Set $\preRotPoly' \exec \preRotPoly \cdot X^{b_{\deg}} \in \R_{\deg, Q}$.} \\
\pcln \text{Run $\acc_{F} \exec [\preRotPoly', 0]^{\top}$} \\
\pcln \text{Run $\acc_{\BR, F} \exec \BlindRotate(\brKey, \acc_{F}, \ct_{\deg}, \vec{u})$.   } \\
\pcln \text{Run $\vec{c}_{Q} \exec \SampleExtract(\acc_{\BR, F}, 1)$.} \\ 
\pcln \text{Run $\vec{c}_{Q, \keySwitchKey} \exec \KeySwitch(\vec{c}_{Q}, \keySwitchKey)$.} \\
\pcln \text{Run $\ct_{\out} \exec \ModSwitch(\vec{c}_{Q, \keySwitchKey}, Q, q)$.} \\
\pcln \text{Return $\ct_{\out}$. } 
} 
} 
\caption{TFHE Bootstrapping.}
\label{figure:tfhe-bootstrapping}
\end{figure}

%
%
%
%
%
%
%
%

\section{Correctness Analysis for Preliminaries}\label{sec:correctness-for-preliminaries}

\begin{proof}[of lemma \ref{lemma:linear-homo-glwe} - Linear Homomorphism of GLWE samples] 
Denote $\vec{c} = [\fr{b}_{\vec{c}}, \vec{a}_{\vec{c}}^{\top}]^{\top}$, where $\fr{b}_{\vec{c}} = \vec{a}_{\vec{c}}^{\top} \cdot \vec{s} + \msg_{\vec{c}} + \fr{e}_{\vec{c}}$ and $\vec{d} = [\fr{b}_{\vec{d}}, \vec{a}_{\vec{d}}^{\top}]^{\top}$, where $\fr{b}_{\vec{d}} = \vec{a}_{\vec{d}}^{\top} \cdot \vec{s} + \msg_{\vec{d}} + \fr{e}_{\vec{d}}$.
Then we have $\vec{c}_{\out} = [\fr{b}, \vec{a}^{\top}]^{\top}  = \vec{c} + \vec{d} = [\fr{b}_{\vec{c}} + \fr{b}_{\vec{d}}, (\vec{a}_{\vec{c}}  + \vec{a}_{\vec{d}})^{\top}]^{\top}$, and
$\fr{b} = \vec{a}^{\top} \cdot \vec{s} + \msg + \fr{e}$, where
$\msg = \msg_{\vec{c}} + \msg_{\vec{d}}$ and $\fr{e} = \fr{e}_{\vec{c}} + \fr{e}_{\vec{d}}$.
Thus, we have that $\Error(\vec{c}_{\out}) = \fr{e}_{\vec{c}} + \fr{e}_{\vec{d}}$,
and $\norm{\Var(\Error(\vec{c}_{\out}))}_{\infty} \leq \bound_{\vec{c}} + \bound_{\vec{d}}$. 

Let $\fr{d} \in \R_{\deg, \basis}$ 
and $\vec{c}_{\out} = \vec{c} \cdot \fr{d} = [\fr{b}_{\fr{d}}, \vec{a}_{\fr{d}}^{\top}]^{\top}$, where
$\vec{a}_{\fr{d}} = \vec{a} \cdot \fr{d}$ and $\fr{b}_{\fr{d}} = \vec{a}_{\fr{d}}^{\top} \cdot \vec{s} + \msg_{\vec{c}} \cdot \fr{d} + \fr{e}_{\vec{c}} \cdot \fr{d}$.

If $\fr{d}$ is a constant monomial with $\norm{\fr{d}}_{\infty} \leq \bound_{\fr{d}}$, then we have
$\norm{\Var(\Error(\vec{c}_{\out}))}_{\infty} \leq \norm{\bound_{\fr{d}}^2 \cdot\Var(\fr{e}_{\vec{c}})}_{\infty} \leq \bound_{\fr{d}}^2 \cdot \bound_{\vec{c}}$.
If $\fr{d}$ is a monomial with its non-zero coefficient uniformly distributed over $[0, \basis-1]$, then
we have $\norm{\Var(\Error(\vec{c}_{\out}))}_{\infty} =  \norm{\frac{1}{3} \cdot \basis^2 \cdot \Var(\fr{e}_{\vec{c}})}_{\infty}$.
This follows from the fact that $\Var(\fr{d}) = \frac{((\basis - 1) - 0 + 1)^2 - 1}{12} = \frac{\basis^2 - 1}{12} \leq \frac{\basis^2}{12}$,
and $\E(\fr{d}) = \frac{\basis - 1 + 0}{2} \leq \frac{\basis}{2}$.
Then for some random variable $X$ with $\E(X) = 0$ 
and assuming $\fr{d}$ and $X$ are uncorrelated we have 
$\Var(\fr{d} \cdot X) = \Var(\fr{d}) \cdot \E(X)^2 + \Var(X) \cdot \E(\fr{d})^2 + \Var(\fr{d}) \cdot \Var(X) \leq \Var(X) \cdot (\frac{\basis^2}{4} + \frac{\basis^2}{12}) \leq  \Var(x) \cdot (\frac{3 \cdot \basis^2}{12} + \frac{\basis^2}{12}) \leq  \Var(x) \cdot \frac{\basis^2}{3}$.

If $\fr{d} \in \R_{\deg, \basis}$ is a constant polynomial with $\norm{\fr{d}}_{\infty} \leq \bound_{\fr{d}}$ then
$\norm{\Var(\Error(\fr{d} \cdot \vec{c}))}_{\infty} \leq \deg \cdot \bound_{\fr{d}}^2  \cdot \bound_{\vec{c}}$.
If $\fr{d} \in \R_{\deg, \basis}$ has coefficients 
distributed uniformly in $[0, \basis-1]$, then
we have $\norm{\Var(\Error(\fr{d} \cdot \vec{c}))}_{\infty} \leq \frac{1}{3} \cdot \deg \cdot \basis^2  \cdot \bound_{\vec{c}}$.
%
\end{proof}

\begin{proof}[of lemma \ref{lemma:linear-homo-ggsw} - Linear Homomorphism of GGSW samples]
The proof follows from the fact that rows of GGSW samples are GLWE samples.
%
\end{proof}

\begin{proof}[of lemma \ref{lemma:correctness-external-product} - Correctness of the External Product]
Denote $\ma{C} = \ma{A} + \msg_{\ma{C}} \cdot \ma{G}_{\ell, \basis, (n+1)} \in \R_{\deg, q}^{(n+1)\ell \times (n+1)}$
and $\vec{d} = [\fr{b}, \vec{a}^{\top}]^{\top} \in \R_{\deg, q}^{(n + 1) \times 1}$.
To ease the exposition, denote $\vec{r} = \Decomp_{\basis, q}(\vec{d}^{t}) \in \R_{\basis}^{1 \times (n+1) \ell}$
and 
$\ma{A} = [\vec{c}_1, \dots, \vec{c}_{(n+1)\ell}]^{\top}$, where for $i \in [(n+1)\ell]$ the vector $\vec{c}_i$
is the transposed GLWE sample of zero making up the row of the matrix $\ma{A}$.
Then from the definition of external product we have $
\vec{r} \cdot \ma{C} = \vec{r} \cdot \ma{A} + \msg_{\ma{C}} \cdot \vec{r} \cdot \ma{G}_{\ell, \basis, (n+1)}$.
From correctness of the gadget decomposition, we have that
$\vec{r} \cdot  \ma{G}_{\ell, \basis, (n+1)} = \vec{d} \in  \R_{\deg, q}^{(n + 1) \times 1}$.
Then, 
$\ma{c}' = \vec{r} \cdot \ma{A} = \sum_{i=1}^{(n+1) \ell} \vec{r}[i] \cdot \vec{c}_i  \in  \R_{\deg, q}^{(n + 1) \times 1}$
is a transposed GLWE sample of zero.
Thus we can write $\vec{c} = \sum_{i=1}^{(n+1) \ell} \vec{r}[i] \cdot \vec{c}_i + \msg_{\ma{C}} \cdot \vec{d}$
and we have $\vec{c}^{\top}$ is a valid GLWE sample of $\msg_{\ma{C}} \cdot \msg_{\vec{d}}$.
Finally $\bound = \norm{\Var(\Error(c))}_{\infty}$ follows from the analysis
of linear combinations of GLWE samples.
In particular, we have
\begin{itemize}
\item $\bound \leq \frac{1}{3}  \cdot \deg  \cdot (n + 1) \cdot \ell \cdot \basis^2 \cdot \bound_{\ma{C}} + \deg \cdot \bound_{\msg_{\ma{C}}}^2 \cdot \bound_{\vec{d}}$ in general, and

\item $\bound \leq \frac{1}{3}   \cdot \deg \cdot (n + 1) \cdot \ell  \cdot \basis^2 \cdot \bound_{\ma{C}} + \bound_{\msg_{\ma{C}}}^2 \cdot \bound_{\vec{d}}$ when $\msg_{\ma{C}}$ is a monomial.
\end{itemize}
%
%
\end{proof}

\begin{proof}[of lemma \ref{lemma:correctness-of-cmux} - Correctness of the $\CMux$ Gate]
Let us first denote
$\vec{d} = \vec{g} - \vec{h} = \GLWE_{\bound_{\vec{g}}}(\vec{s}, \msg_{\vec{g}}) - \GLWE_{\bound_{\vec{h}}}(\vec{s}, \msg_{\vec{h}}) = \GLWE_{\bound_{\vec{d}}}(\vec{s}, \msg_{\vec{d}})$,
where $\msg_{\vec{d}} = \msg_{\vec{g}} - \msg_{\vec{h}}$.
Then from correctness of the external product and linear homomorphism of GLWE samples
we have
\begin{align*}
\vec{c}_{\out} &= \extProd(\ma{C}, \vec{d}) + \vec{h} \\
 &= \extProd(\GGSW_{\bound_{\ma{C}}, n, \deg, q, \basis}(\vec{s}, \msg_{\ma{C}}), \\
  &~~+ \GLWE_{\bound_{\vec{d}}}(\vec{s}, \msg_{\vec{d}})) + \GLWE_{\bound_{\vec{h}}}(\vec{s}, \msg_{\vec{h}}) \\
  &= \GLWE_{\bound}(\vec{s}, \msg_{\out}),
\end{align*}
where $\msg_{\out} = \msg_{\ma{C}} \cdot (\msg_{\vec{g}} - \msg_{\vec{h}}) + \msg_{\vec{h}}$.
Hence if $\msg_{\out} = 0$, we have $\msg_{\out} = \msg_{\vec{h}}$,
and if $\msg_{\out} = 1$, we have $\msg_{\out} = \msg_{\vec{g}}$.

For the error analysis let us denote explicitly 
$\Error(\vec{g}) = \fr{e}_{\vec{g}}$ and $\Error(\vec{h}) = \fr{e}_{\vec{h}}$.
Thus we also have $\Error(\vec{d}) = \fr{e}_{\vec{d}} = \fr{e}_{\vec{g}} = \fr{e}_{\vec{h}}$.
We also denote $\Error(\ma{C}[i]^{\top}) = \fr{e}_i$
  for $i \in [(n+1) \cdot \ell]$ and $\ell = \roundUp{\log_{\basis}{q}}$,
and $\Decomp_{\basis, q}(\vec{d}) = \vec{r}$.
Then, from the correctness analysis
of the external product
we have that  
\begin{align*}
\Error(\vec{c}_{\out}) &= \sum_{i=1}^{(n+1)\cdot \ell} \vec{r}[i] \cdot \fr{e}_i + \msg_{\ma{C}} \cdot \fr{e}_{\vec{d}} + \fr{e}_{\vec{h}} \\
&= \sum_{i=1}^{(n+1)\cdot \ell} \vec{r}[i] \cdot \fr{e}_i + \msg_{\ma{C}} \cdot \fr{e}_{\vec{g}} - \msg_{\ma{C}} \cdot \fr{e}_{\vec{h}} + \fr{e}_{\vec{h}}.
\end{align*}
Thus if we have $\msg_{\ma{C}} = 0$, then
$\Error(\vec{c}_{\out}) = \sum_{i=1}^{(n+1)\cdot \ell} \vec{r}[i] \cdot \fr{e}_i +  \fr{e}_{\vec{h}}$,
and if $\msg_{\ma{C}} = 1$, we obtain
$\Error(\vec{c}_{\out})= \sum_{i=1}^{(n+1)\cdot \ell} \vec{r}[i] \cdot \fr{e}_i + \fr{e}_{\vec{g}}$.

Finally, we can place an upper on the variance
$\bound \leq \frac{1}{3} \cdot (n+1) \cdot \deg \cdot \ell \cdot \basis^2 \cdot \bound_{\ma{C}} + \max(\bound_{\vec{d}}, \bound_{\vec{h}})$.
\end{proof}

\begin{proof}[of lemma \ref{lemma:correctness-modulus-switching} - Correctness of Modulus Switching]
Denote $\vec{c} = (b, \vec{a})$, where 
$b = \vec{a}^{\top} \cdot \vec{s} + \Delta_{Q, t}  \cdot \msg + e \in \ZZ_Q$
and $\norm{\Var(e)}_{\infty} \leq \bound_{\vec{c}}$.
From the assumption that $Q = 0 \mod t$ we have $\Delta_{Q, t} = \round{\frac{Q}{t}} = \frac{Q}{t}$.
Then we have that the following:
\begin{align*}
\Phase(\round{\frac{q}{Q} \cdot \vec{c}}) &= \round{\frac{q}{Q} \cdot b} - \round{\frac{q}{Q} \cdot \vec{a}^{\top}} \cdot \vec{s} \\
&= \frac{q}{Q} \cdot b + r - \frac{q}{Q} \cdot \vec{a}^{\top} \cdot \vec{s} + \vec{r}^{\top} \cdot \vec{s} \\
&=  \frac{q}{Q} \cdot \Delta_{Q, t}  \cdot \msg +  \frac{q}{Q} \cdot e + r +  \vec{r}^{\top} \cdot \vec{s} \\
&= \frac{q}{t} \cdot \msg +  \frac{q}{Q} \cdot e + r +  \vec{r}^{\top} \cdot \vec{s}  
\end{align*}
where $r \in \RR$ and $\vec{r} \in \RR^{n}$ are in $[-\frac{1}{2}, \frac{1}{2}]$.
Furthermore, $r$ and $\vec{r}$ are close to uniform distribution over their support.
Denote as $t \in \ZZ$  a uniformly random variable over $\{-1, 0, 1\}$. 
Clearly we have that $\Var(r) \leq \Var(t)$ and similarly for all entries of $\vec{r}$.
Recall that $\norm{\Var(t)}_{\infty} = \frac{2}{3}$ and $\norm{\E(t)}_{\infty} = 0$.
Let us denote $\hamming(\vec{s}) = h$.
Therefore, we have
\begin{align*}
\norm{&\Var(\Error(\round{\frac{q}{Q} \cdot \vec{c}})))}_{\infty} \\
 &= \norm{\Var(\frac{q}{Q} \cdot e + r +  \vec{r}^{\top} \cdot \vec{s})}_{\infty} \\
&= \norm{\Var(\frac{q}{Q} \cdot e)}_{\infty} + \norm{\Var(r)}_{\infty} +  \norm{\Var(\vec{r}^{\top} \cdot \vec{s}))}_{\infty} \\
&\leq \frac{q^2}{Q^2} \cdot \bound_{\vec{c}} + \frac{2}{3} + \sum_{i=1}^{\hamming(\vec{s})} \norm{\Var(\vec{r}[i] \cdot \vec{s}[i])}_{\infty} \\
&\leq \frac{q^2}{Q^2} \cdot \bound_{\vec{c}} + \frac{2}{3} + \frac{2}{3} \cdot \hamming(\vec{s}) \cdot \big(\norm{\Var(\vec{s})}_{\infty} + \norm{\E(\vec{s})}_{\infty}^2\big) \\
\end{align*}
The last inequality follows from the
fact that for all $i \in [\hamming(\vec{s})]$ we have
\begin{align*}
\norm{&\Var(\vec{r}[i] \cdot \vec{s}[i])}_{\infty} \\
 &\leq \norm{\big(\E(\vec{s}[i])^2 \cdot \Var(t) + \E(t)^2 \cdot \Var(\vec{s}[i]) \\
 &~~+ \Var(\vec{s}[i]) \cdot \Var(t) \big)}_{\infty} \\
&= \norm{\big(\E(\vec{s}[i])^2 \cdot \frac{2}{3} + \Var(\vec{s}[i]) \cdot \frac{2}{3} \big)}_{\infty} \\
&= \frac{2}{3} \cdot \big(\norm{\E(\vec{s})}_{\infty}^2  + \norm{\Var(\vec{s})}_{\infty}\big).
\end{align*}
In the above we assume that $\norm{\Var(\vec{s}[i])}_{\infty} = \norm{\Var(\vec{s})}_{\infty}$ for all $i \in [n]$.

%
%
%
\end{proof}

\begin{proof}[of lemma \ref{lemma:correctness-sample-extraction} - Correctness of Sample Extraction]
Denote $\vec{s} = \coefs(\fr{s}) \in \ZZ_q^{\deg}$ and $b' = \coefs(\fr{b})[k] \in \ZZ_q$
and $\vec{a} = \coefs(\fr{a}) \in \ZZ_q^{\deg}$.
Denote $\fr{b} = \fr{a} \cdot \fr{s} + \msg + \fr{e} \in \R_{\deg, q}$,  $\msg = \sum_{i=1}^{\deg} \msg_i \cdot X^{i-1}$
and $\fr{e} = \sum_{i=1}^{\deg} e_i \cdot X^{i-1}$
then it is easy to see, that $b' = \coefs(\fr{a} \cdot \fr{s})[k] + \coefs(\msg)[k] + \coefs(\fr{e}) = \coefs(\fr{a} \cdot \fr{s})[k] + \msg_k + e_k$.
Furthermore, denote $\fr{a} = \sum_{i=1}^{\deg} a_i \cdot X^{i-1}$ and  $\fr{s} = \sum_{i=1}^{\deg} s_i \cdot X^{i-1}$.
Denote $\fr{s} \cdot \fr{a} = (\sum_{i=1}^{\deg} a_i \cdot X^{i-1}) \cdot (\sum_{i=1}^{\deg} s_i \cdot X^{i-1})$.
By expanding the product we have that the $k$-th coefficient of $\fr{s} \cdot \fr{a}$ is given by 
$\sum_{i=1}^{\deg} s_i \cdot \NegSgn(k - i + 1) \cdot  a_{(k - i \mod \deg) + 1}$. Note that
when working over the ring $\R_{\deg, q}$ defined by the cyclotomic polynomial $\cyclo_{\deg} = X^{\deg} + 1$,
we have that for $(k - i \mod \deg) + 1 \leq 0$, the sign of $a_{(k - i \mod \deg) + 1}$ is reversed,
thus we multiply the coefficient with $\NegSgn(k - i + 1)$.
Now if we set $\vec{a}'[i] = \NegSgn(k - i + 1) \cdot  a_{(k - i \mod \deg) + 1}$,
then the extracted sample is a valid LWE sample of $\msg_k$ with respect to
key $\vec{s}$.
Finally, since 
the error
in $e_k$ is a single coefficient of $e$, we have that $\bound_{\ct} = \bound$. 
\end{proof}

\begin{proof}[of lemma \ref{lemma:correctness-lwe-to-rlwe-packing} - Correctness of Key Switching]
Let us first note that for all $i \in [n]$ we have
\begin{align*}
[b_i, \vec{a}_i^{\top}]^{\top} &= \sum_{j=1}^{\ell_{\keySwitchKey}} \ma{A}[i,j] \cdot \keySwitchKey[i,j] \\
&= \sum_{j=1}^{\ell_{\keySwitchKey}} \ma{A}[i,j] \cdot \GLWE_{\bound_{\keySwitchKey}, n', \deg, q}(\vec{s}', \vec{s}[i] \cdot \basis_{\keySwitchKey}^{j}) \\
&=  \GLWE_{\bound_{1}, n', \deg, q}(\vec{s}', \vec{a}[i] \cdot \vec{s}[i]),
\end{align*}
where  
$\bound_1 \leq (\frac{1}{3}) \cdot \ell_{\keySwitchKey} \cdot \basis_{\keySwitchKey}^2 \cdot \bound_{\keySwitchKey}$.
Then  note that
\begin{align*}
[b', \vec{a}'^{\top}]^{\top} &= \sum_{i=1}^{n} \GLWE_{\bound_{1}, n', \deg, q}(\vec{s}',  \vec{a}[i] \cdot \vec{s}[i]) \\ 
 &=  \GLWE_{\bound_{2}, n', \deg, q}(\vec{s}', \sum_{i=1}^{n} \vec{a}[i] \cdot \vec{s}[i])  \\
 &= \GLWE_{\bound_{2}, n', \deg, q}(\vec{s}', \vec{a}^{\top} \cdot \vec{s}),
\end{align*}
where  
$\bound_{2} \leq n \cdot \bound_1 \leq  (\frac{1}{3}) \cdot n \cdot \ell_{\keySwitchKey} \cdot \basis_{\keySwitchKey}^2 \cdot \bound_{\keySwitchKey}$.
Let us denote $b = \vec{a}^{\top} \cdot \vec{s} + \msg + e$ and
$b' = \vec{a}'^{\top} \cdot \vec{s} + \vec{a}^{\top} \cdot \vec{s} + e'$,
then
\begin{align*}
\vec{c}_{\out} &= [b, 0]^{\top} - [b', \vec{a}'^{\top}]^{\top} \\
 &= [b - b', -\vec{a}'^{\top}]^{\top} \\
 &= [-\vec{a}'^{\top} \cdot \vec{s}' + \msg + e - e'].
\end{align*}
Hence, $\vec{c}_{\out}$ is a valid GLWE sample of $\msg$ with respect to key $\vec{s}'$ 
and 
\begin{align*}
\norm{\Var(\Error(\vec{c}))}_{\infty} &\leq \bound \\
 &\leq \norm{e - e'}_2 \leq \bound_{\vec{c}} + \bound_{2} \\
  &\leq \bound_{\vec{c}} + (\frac{1}{3}) \cdot n \cdot \ell_{\keySwitchKey} \cdot \basis_{\keySwitchKey}^2 \cdot \bound_{\keySwitchKey}.
\end{align*}
 
\end{proof}

\begin{proof}[of lemma \ref{lemma:correctness-tfhe-blind-rotate} - Correctness of TFHE-Style Blind-Rotation]
Note that the blind rotation key $\brKey$ consists of samples of bits:
$\RGSW_{\bound_{\BR}, \deg, Q, \basis_{\RGSW}}(\fr{s}, \ma{Y}[i,j])$, where $\ma{Y}[i,j] \in \BB$ for $i \in [n]$ 
and $j \in [u]$.
Denote the current accumulator as $\acc_{curr} \in \RLWE_{\bound_{\curr}, \deg, Q}(\fr{s}, \msg_{\curr})$ for some $\msg_{\curr} \in \R_{\deg, Q}$.
If we increment 
$\acc_{\nextSym} \exec \CMux(\brKey[i, j], \acc_{\curr} \cdot X^{-\vec{a}[i] \cdot \vec{u}[j]}, \acc_{\curr})$
for some $i \in [n]$ and $j \in [u]$, then from correctness of the CMux gate 
we have that $\acc_{\nextSym} \in \RLWE_{\bound_{\nextSym}, \deg, Q}(\fr{s}, \msg_{\nextSym})$,
where $\msg_{\nextSym} = \msg_{\curr} \cdot X^{-\vec{a}[i] \cdot \vec{u}[j] \cdot K[i, j] \mod 2 \cdot \deg} \in \R_{\deg, Q}$.
Overall, we can represent the last message $\msg_{\out}$ as
\begin{align*}
\msg_{\out} &= \msg_{\acc} \cdot \prod_{i=1}^n \prod_{j=1}^{u}  X^{-\vec{a}[i] \cdot \vec{u}[j] \cdot K[i, j]} \\
 &=  \msg_{\acc} \cdot \prod_{i=1}^n X^{-\vec{a}[i] \cdot \vec{s}[i]} \\
 &= \msg_{\acc} \cdot X^{- \vec{a}^{\top} \cdot \vec{s} \mod 2 \cdot \deg}.
\end{align*}

From the analysis of the $\CMux$ gate we have that 
$\bound_{\nextSym} \leq \bound_{\curr} +  (\frac{2}{3}) \cdot \deg \cdot \ell_{\RGSW} \cdot \basis_{\RGSW}^2 \cdot \bound_{\brKey}$.
  
Finally since we perform $n \cdot u$ iterations,
we have that
\begin{align*} 
\bound_{\out} \leq  \bound_{\acc} +  (\frac{2}{3}) \cdot  n \cdot u \cdot \deg \cdot \ell_{\RGSW} \cdot \basis_{\RGSW}^2 \cdot \bound_{\brKey}.
\end{align*} 
\end{proof}

\begin{proof}[of theorem \ref{thm:correctness-tfhe-bootstrapping} - Correctness of TFHE Bootstrapping] 
Let us denote  $\Phase(\ModSwitch(\ct, q, 2 \deg)) = \msg' = \Delta_{2 \cdot \deg, t} \cdot \msg + e \mod 2 \deg$. 
From correctness of modulus switching we have $\norm{\Var(e)}_{\infty} \leq \bound_{\deg} \leq (\frac{2 \cdot \deg}{q})^2 \cdot \bound_{\ct} + \frac{2}{3}  + \frac{2}{3} + \hamming(\vec{s}) \cdot (\norm{\Var(\vec{s})}_{\infty} + \norm{\E(\vec{s})}_{\infty}^2)$.
From the correctness of blind rotation we have 
$\acc_{\BR, F} \in \RLWE_{\bound_{\BR}, \deg, Q}(\fr{s}, \preRotPoly \cdot X^{\msg' \mod 2 \cdot \deg})$, where $\bound_{\BR}$ is
the bound induced by the blind rotation algorithm.
From the assumption on $\preRotPoly$, given that $\msg' \leq \deg$ and the correctness of
sample extraction we have
$\vec{c}_{Q} \in \LWE_{\bound_{\vec{c}, Q}, \deg, Q}(\vec{s}_{F}, \Delta_{Q, t} \cdot F(\msg))$, where
$\bound_{\vec{c}, Q} \leq \bound_{\BR}$.

From correctness of key switching
we have that $\vec{c}_{Q, \keySwitchKey} \in \LWE_{\bound, n, q}(\vec{s}, \Delta_{Q, t} \cdot F(\msg))$,
where 
 $\bound_{\vec{c}, Q, \keySwitchKey} \leq \bound_{\vec{c}, Q} + \bound_{KS}$, where $\bound_{KS}$ is the bound
 induced by the key switching procedure.
From correctness of modulus switching we have that 
$\ct_{\out} \in \LWE_{\bound_{\out}, \deg, q}(\vec{s}, \Delta_{q, t} \cdot F(\msg))$,
where  
\begin{align*}
\bound_{\out} \leq& (\frac{q^2}{Q^2}) \cdot \bound_{\vec{c}, Q, \keySwitchKey} \\
 &+ \frac{2}{3} + \frac{2}{3} \cdot \hamming(\vec{s}_{F}) \cdot  (\norm{\Var(\vec{s}_{F})}_{\infty} + \norm{\E(\vec{s}_{F})}_{\infty}^2).
\end{align*}

To summarize we have:
\begin{align*}
\bound_{\out} \leq& (\frac{q^2}{Q^2}) \cdot  (\bound_{\BR} + \bound_{KS}) \\
 &+ \frac{2}{3} + \frac{2}{3} \cdot \hamming(\vec{s}_{F}) \cdot  (\norm{\Var(\vec{s}_{F})}_{\infty} + \norm{\E(\vec{s}_{F})}_{\infty}^2), 
\end{align*}
where
\begin{align*}
\bound_{\BR} \leq (\frac{2}{3}) \cdot  n \cdot u \cdot \deg \cdot \ell_{\RGSW} \cdot \basis_{\RGSW}^2 \cdot \bound_{\brKey}, \text{ and} \\
\bound_{KS} \leq  (\frac{1}{3}) \cdot  \deg \cdot \ell_{\keySwitchKey} \cdot \basis_{\keySwitchKey}^2 \cdot  \bound_{\keySwitchKey}.
\end{align*}
\end{proof}

%

\end{document}